\documentclass[letterpaper,12pt,peerreviewca,draftcls]{IEEEtran}
\usepackage{csm16}
\usepackage[margin=1in]{geometry}
\usepackage{amsmath} 

\usepackage[utf8]{inputenc}
\usepackage[english]{babel}
\usepackage{setspace}
\usepackage{cite}
\usepackage{multirow}
\usepackage[table,xcdraw]{xcolor}
\usepackage{epsfig} 
\usepackage{graphics} 
\usepackage{graphicx}
\usepackage[mathscr]{euscript}

\usepackage{tikz}
\usetikzlibrary{intersections}
\usetikzlibrary{decorations.pathreplacing}

\usepackage{multicol}
\usepackage{algorithm}
\usepackage[noend]{algpseudocode}

\usepackage{amsthm} 
\usepackage{amsmath} 
\usepackage{amssymb}  
\usepackage[percent]{overpic}

\theoremstyle{plain}
\newtheorem{Theorem}{Theorem}

\newtheorem{Corollary}{Corollary}

\usepackage{blkarray}
\usepackage{makecell}

\usepackage[customcolors]{hf-tikz}

\tikzset{style green/.style={
    set fill color=green!50!lime!60,
    set border color=white,
  },
  style cyan/.style={
    set fill color=cyan!90!blue!60,
    set border color=white,
  },
  style gray/.style={
    set fill color=gray!20,
    set border color=white,
  },
  style orange/.style={
    set fill color=orange!80!red!60,
    set border color=white,
  },
  hor/.style={
    above left offset={-0.15,0.31},
    below right offset={0.15,-0.125},
    #1
  },
  ver/.style={
    above left offset={-0.1,0.3},
    below right offset={0.15,-0.15},
    #1
  }
}

\usepackage{url}
\usepackage{graphicx,xcolor}
\usepackage{verbatim}
\makeatletter
\newcommand{\verbatimfont}[1]{\def\verbatim@font{#1}}%
\makeatother
\verbatimfont{\ttfamily\small}

\newcommand{\bi}{\begin{itemize}}\newcommand{\ei}{\end{itemize}}
\newcommand{\be}{\begin{equation}}\newcommand{\ee}{\end{equation}}
\newcommand{\bee}{\begin{enumerate}}\newcommand{\eee}{\end{enumerate}}
\newcommand{\bea}{\begin{eqnarray}}\newcommand{\eea}{\end{eqnarray}}
\newcommand{\beas}{\begin{eqnarray*}}\newcommand{\eeas}{\end{eqnarray*}}
\newcommand{\bc}{\begin{center}}\newcommand{\ec}{\end{center}}
\usepackage[left,pagewise]{lineno} 
\usepackage[english]{babel} 
\usepackage{blindtext}

\title{Model Based Control of Soft Robots: \\ {\LARGE A Survey of the State of the Art and Open Challenges}}

\author{Cosimo Della Santina, Christian Duriez, Daniela Rus}

\newif\ifPDF \ifx\pdfoutput\undefined\PDFfalse \else\ifnum\pdfoutput > 0\PDFtrue \else\PDFfalse \fi \fi
\ifPDF 
\usepackage[pdftex, plainpages = false, colorlinks=true, linkcolor=black, citecolor = green!50!blue, urlcolor = blue, filecolor=black, pagebackref=false, hypertexnames=false,  pdfpagelabels ]{hyperref}
\fi

\usepackage[margin=1in]{geometry}

\begin{document}
\maketitle

\begin{center}
	What follows is an \textbf{\color{red} old version} of the manuscript. You can find the \textbf{\color{blue} final one }at the following link:  \url{https://www.dropbox.com/scl/fi/ith1d3golzovg4tkusihg/Model-Based_Control_of_Soft_Robots_A_Survey_of_the_State_of_the_Art_and_Open_Challenges.pdf?rlkey=yvrf7ruvnr21i0e3h3m6wd0l3&dl=0}
	
	Feel free to get in touch via \texttt{c.dellasantina@tudelft.nl} if you cannot get access.
\end{center}

\CSMsetup

From a functional standpoint, classic robots are not at all similar to biological systems. If compared with rigid robots, animals' body looks overly redundant, imprecise, and weak. Nevertheless, animals can still perform a vast range of activities with unmatched effectiveness. Many studies in bio-mechanics have pointed to the elastic and compliant nature of the muscle-skeletal system as a fundamental ingredient explaining this gap.
Thus, to reach performance comparable to the natural ones, elastic elements have been introduced in rigid bodied robots leading to articulated soft robotics \cite{della2021soft}. In continuum soft robotics, this concept is brought to an extreme. Here, softness is not concentrated at the joint level but instead distributed across the whole structure. As a result, soft robots (from now on, we will omit the adjective \textit{continuum}) are entirely made of continuously deformable elements. This design solution aims to bring robots closer to invertebrate animals and soft appendices of vertebrate animals (e.g., an elephant's trunk, the tail of a monkey). Several soft robotic hardware platforms have been proposed, with increasingly higher reliability and functionalities. In this process, considerable attention has been devoted to the technological side of the problem, leading to a large assortment of hardware solutions. In turn, this abundance opened up to the challenge of developing effective control strategies that can manage the soft body and exploit its embodied intelligence.

Historically and across many application domains, model-based techniques are the first advanced control algorithms to appear and substitute heuristic rules. Data-driven and machine learning approaches usually come later when moving to more extreme control scenarios. This has also been the case for standard robotics, whose history proceeded parallel to the development of control theory: from the frequency domain to linear state space control, to fully nonlinear domain, and only recently to machine learning. Vice versa, the development of control algorithms in soft robotics has followed a reversed path. In the early days, machine learning strategies have been the way to control soft robots - except for the quasi-static and purely kinematic scenarios. Indeed, it has been long believed that model-based strategies were unfeasible for the soft robotic application due to the large variability of technological solutions and the overwhelming complexity of the modeling task.

Over the past few years, two main factors have been challenging this view.  First, theoretical and experimental investigations have shown that feedback schemes are robust to rough approximations of the soft robot dynamics. Interestingly, even vastly simplified descriptions already provide enough information to improve the performance significantly compared to the model-free baseline. Second, a new wave of finite-dimensional modeling techniques tailored to soft robots has appeared in the literature, which are simultaneously accurate, manageable, and interpretable.  Even if a complete application to closed-loop control has yet to come for some of these models, these theoretical works identify an underlying mathematical structure in soft robotics. Therefore, they lay a solid ground on which to study the control problem.

This work aims to introduce the control theorist perspective to this novel development in robotics. We aim to remove the barriers to entry into this field by presenting existing results and future challenges using a unified language and within a coherent framework. Indeed, the main difficulty in entering this field is the wide variability of terminology and scientific backgrounds, making it quite hard to acquire a comprehensive view on the topic. Another limiting factor is that it is not obvious where to draw a clear line between the limitations imposed by the technology not being mature yet and the challenges intrinsic to this class of robots. In this work, we consider as intrinsic the continuum or multi\--body dynamics, the presence of a non\--negligible elastic potential field, and the variability in sensing and actuation strategies. The hystereses and non\--ideal behaviors affecting sensors, actuators, and main body are considered relevant but not intrinsic - since we believe that with the advance of the technology, these aspects should be overcome.
Of the many review papers about soft robotics \cite{trivedi2008soft,pfeifer2012challenges,kim2013soft,rus2015design,laschi2016soft,calisti2017fundamentals,cianchetti2018biomedical,wang2018toward,chen2020design,della2021soft}, only \cite{george2018control} is focused on the control challenge, which is, however, not focused on the model-based approach. 

%
\section{Finite Dimensional Models for Control Purposes}

In its exact formulation, continuum soft robots belong to the domain of continuum mechanics. As such, their dynamics is formulated as an infinite-dimensional system, i.e., via partial differential equations (PDEs). Yet, recent work has clearly shown that finite-dimensional approximations of the robot's dynamics can be formulated that assume the form of standard ordinary differential equations (ODEs). These formulations are simultaneously tractable and precise enough to describe the soft robot behavior with the necessary precision. Contrary to the rigid case, developing models is an integral part of the control design process in soft robotics. Usual models of rigid robots can serve as a base for simulating and controlling these systems. Instead, with soft robots, simulation models and algorithm design models come from different assumptions and approximations. The former must be accurate, possibly at the cost of computational efficiency and simplicity of interpretation. In contrast, the latter must be lower-dimensional. They must capture the core essence of the dynamics - possibly neglecting the finer details - and they must land themselves to be used for formally assessing the structural properties of the robot and the closed-loop behavior. The reader interested in skipping the details about modeling and jumping directly to the control part is advised to still read the subsections titled \textit{Finite-dimensional approximations} and \textit{Existence of equilibria}.

\subsection{When the rigid part is dominant}
It is not uncommon to find full-fledged soft robotic technologies integrated into essentially rigid structures. This is for example the case of the soft neck of a rigid humanoid robot discussed in \cite{deutschmann2017position,munoz2020iso}, or the soft muscles  actuating rigid links \cite{niiyama2007mowgli,brochu2012dielectric,kellaris2021spider}. In all these cases, the dynamics of the rigid part is essentially dominant with respect to the soft continuum part. Thus, the system model can be obtained with a good level of approximation by applying standard multi-body dynamics machinery to the rigid portion and accounting for the soft part through a nonlinear lumped impedance. 

\begin{figure}[t!]
	\centering
	\includegraphics[height=0.45\textwidth]{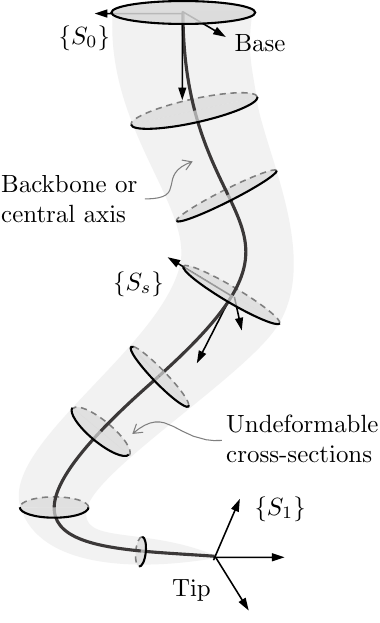}
	\hspace{0.1\columnwidth}
	\includegraphics[height=0.45\textwidth]{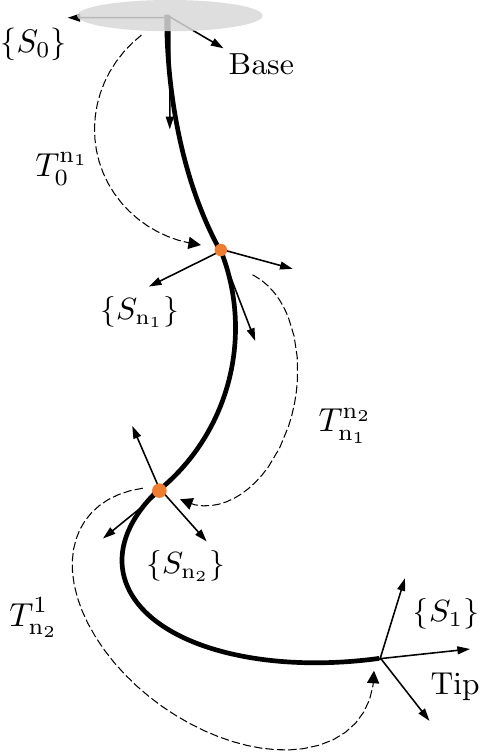}
	\caption{\small The left side shows a soft robot as described within rod theories. The central axis or backbone is a spatial curve that can deform. Cross-sections are assumed undeformable, and they are rigidly connected to the curve. Some examples of disk-like cross-sections are highlighted in the figure. To each point is also attached a reference frame $\{S_{s}\}$, where $s$ is the normalized arclength. The left side of the picture reports an example of piecewise variable strain models with three segments. First, four nodes along the backbone are identified ($0,\mathrm{n}_1,\mathrm{n}_2,1$). Then, the associated transformation matrices ($T_0^{\mathrm{n}_1},T_{\mathrm{n}_1}^{\mathrm{n}_2},T_{\mathrm{n}_2}^1$) are parametrized by a finite set of variables. The robot configuration $q$ is defined as the collection of these variables. In Piecewise Constant Curvature models, two subsequent nodes are always connected by an arc of a circle. 
	\label{fig:rod_models}}
\end{figure}

\begin{figure}
    \centering
    \includegraphics[width=.6\textwidth]{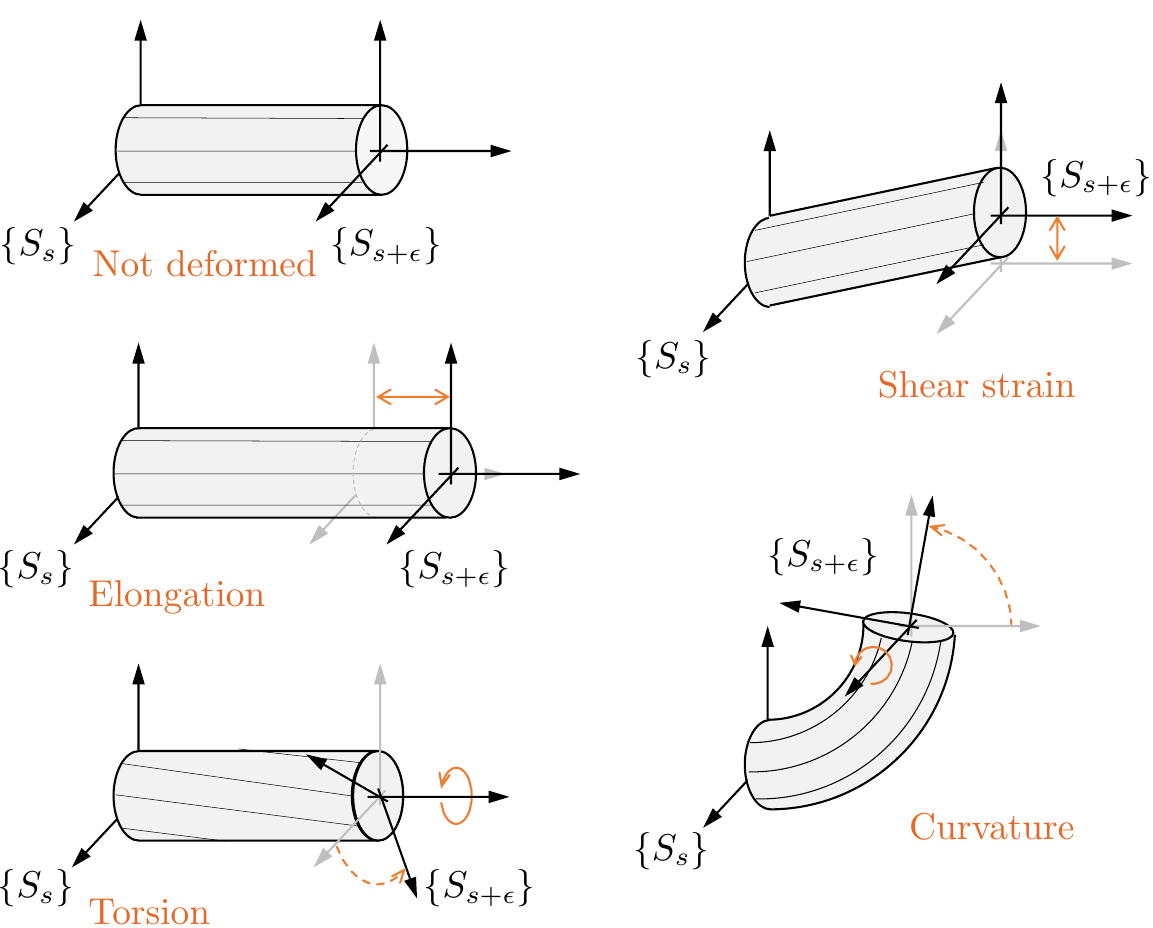}
    \caption{\small  Representation of the six pure strains, corresponding to a $\xi(s) \in \mathbb{R}^6$ with all elements but one equal to zero. Each one is associated with either a pure translation or rotation along with one of the three local axes of the reference frame ${S_{s}}$. Infinitesimal translations are referred to as elongation when occurring along the axis tangent to the backbone and as shear strain if happening along with the two orthogonal directions (only one shown in the picture). Similarly, infinitesimal rotations are referred to as torsion when happening along the axis tangent to the backbone and as curvature if occurring in the direction of the two orthogonal directions (only one shown in the picture). }
    \label{fig:deformations}
\end{figure}

\subsection{Rod Models}

Many soft robots have one physical dimension longer than the other two, such as the tentacle-like system shown by the left-hand side of Fig. \ref{fig:rod_models}. Their whole configuration can be effectively approximated by neglecting volumetric deformations and focusing on the behavior of their central axis. This assumption works well in practice and allows to rely on well-established theories in continuum mechanics of rods.

Consider a continuum and infinitely thin element (rod from now on) of length $L$, as, for example, the one in Fig. \ref{fig:rod_models}. Its posture in space can be completely described by the spatial curve $x: [0,1] \times \mathbb{R} \rightarrow SE(3)$. The domain $[0,1]$ is the normalized arc length of the rod, and $SE(3)$ is the Special Euclidean group of dimension $3$. The element $x(s,t)$ is the full posture at time $t$ of the infinitesimal element of the rod that is at a distance $s L$ from the robot's base. The total length of the rod is $L$. So, $x(t,0)$ is the configuration of the base, $x(t,1)$ of the tip, $x(t,1/2)$ of the middle point, and so on. For simplicity of notation, we will assume $SE(3)$ to be parametrized through $\mathbb{R}^6$, for example, by using Euler angles. However, the reader must be warned that many of complexities in advanced rod models come from dealing with this parametrization properly.
To each point $s$ along the rod, we can associate a mass density $m(s) \in \mathbb{R}^+$, an external load in the form of a generic wrench $f(s,t) \in \mathbb{R}^{6}$, and a velocity $\dot{x}(s,t) \in \mathbb{R}^{6}$ . The total mass of the segment is $\int_0^1 m(s) \mathrm{d} s$.
The configuration is described by the local strains: curvatures, twist, elongation, and shear. These are the variations of $x$ for the infinitesimal element $s$ with respect to the previous one. As such, they are a function $\xi:[0,1]  \rightarrow \mathbb{R}^6$ (or to a smaller space in case some of the deformations are not considered). Visualizations of pure strains are provided in Fig. \ref{fig:deformations}.
The rod posture $x$ can always be recovered from the strains $\xi$ by integration. The result is a continuous version of what in classic robotics would be regarded as the forward kinematics of the robot.

It is worth stressing that, despite the formulation referring to an infinitely thin structure (the rod), these models can and are commonly applied every time a central axis can be identified. In this case, we say that the curve $x$ is the backbone of the soft robot. Under the assumption that the cross-section of the soft robot changes negligibly during deformations, the contribution of the area can be included by associating an infinitesimal rotational inertia $\mathcal{J}(s) \in \mathbb{R}^{3 \times 3}$ to each point along the rod.

At this point, the main ingredients necessary to describe a soft robot within the rod modeling framework have been laid down. The following subsections will focus on reviewing alternative solutions for formulating the dynamics of these systems. We will start from exact infinite formulations, then quickly move to survey the various existing alternative to introducing expert intuitions into the problem and getting to a finite-dimensional model of the robot that can be used for control purposes. We will often refer to Piecewise Constant Curvature models when providing examples. This is done for simplicity and because most of the properties of more complex models are already present in this more straightforward and widespread solution.

\subsubsection{Infinite Dimensional Models}
Thanks to their capability of exactly describing continuum structures, and in virtue of their solid theoretical foundations, Kirchhoff-Clebsch-Love and Cosserat rod theories \cite{coleman1993dynamics,spillmann2007corde,lang2011multi,gazzola2018forward} are a natural choice for describing soft robots having rod-like structures (see Fig. \ref{fig:rod_models}). Leveraging these frameworks, the statics and dynamics of tendon actuated continuum and soft robots are derived and experimentally validated in \cite{trivedi2008geometrically,rucker2011statics} and \cite{renda2014dynamic} respectively. Multiple Cosserat-rod models can be combined together through coupled boundary conditions to describe the kinematics of parallel soft robots \cite{bryson2014toward,black2017parallel}. 
A tutorial on the dynamic Cosserat model for tendon-driven continuum robots is provided in \cite{janabi2021cosserat}. These models have infinite-dimensional states, and as such, they are formulated as PDEs. This makes it hard to use them directly for control (see sidebar \ref{sb:infinite_dimensional}xx for more details). Nonetheless, they can be profitably used to extract steady-state solutions. The use of the Magnus expansion to solve the kinematics of Cosserat rods is discussed in \cite{renda2020geometric, orekhov2020solving}. Also, the direct application to simulation is arduous but not impossible. For example, \cite{till2019real} performs a time discretization, which transforms the PDE into an ODE in the $s$ variable only. The latter is then solved at every time step to find the robot's shape.  Nonlinear observers can be used to speed up the convergence \cite{thamo2021rapid}.

\subsubsection{Finite dimensional approximations}
The alternative to PDE formulations is to restrict the range of possible strains $\xi$ to a finite-dimensional functional space. Two classes of strategies exist to achieve this goal: piecewise constant strain models and functional parametrizations. Both of them will be discussed in detail below. At the current stage, what is essential to keep in mind is that using these techniques, the strain $\xi$ can be approximated as a function of the vector $q \in \mathbb{R}^{n}$ that serves as the configuration of the soft robot. This critical step enables the recasting of concepts from classic discrete robotics to the new continuum context. For a start, the kinematics of a soft robot can now be defined as follows
\begin{equation}\label{eq:kinematics}
	\dot{x}(s,q(t)) = J(s,q(t))\,\dot{q}(t), \quad J(s,q) = \frac{\partial h(s,q)}{\partial q},
\end{equation}
where $h(s,q) \in \mathbb{R}^6$ is the map - called forward kinematics - connecting the configuration $q(t)$ to the posture $x(s,t)$ for each point $s$ along the backbone. The matrix-valued function $J$ is the Jacobian of $h$. The following set of ODEs can be directly derived from \eqref{eq:kinematics} via standard Lagrangian mechanics machinery
\begin{equation}\label{eq:dynamics}
	\underbrace{M(q)\ddot{q} + C(q,\dot{q})\dot{q} + G(q)}_{\text{Multi-body dynamics}} + \underbrace{D(q)\dot{q} + K(q)}_{\text{Elastic and dissipative forces}} = \underbrace{A(q) \tau}_{\text{Model of underactuation}},
\end{equation}
where $(q,\dot{q})$ forms the robot state. 

The inertia matrix $M(q) \in \mathbb{R}^{n \times n}$ is evaluated as follows
\begin{equation}\label{eq:inertia}
M(q) = \int_{0}^{1} J^{\top}(q,s) \begin{bmatrix}
m(s) I &0 \\
0 & \mathcal{J}(s)
\end{bmatrix} J (q,s) \; \mathrm{d} s \succ 0,
\end{equation}
where $m(s)$ and $\mathcal{J}(s)$ are the mass and inertia distributions respectively. 
Note that $M(q)$ may become singular in some configurations. This however can always be avoided by properly parametrizing the configurations space.
As for a rigid robot, the inertia matrix verifies
\begin{equation}
||M(q)|| \leq c_{\mathrm{m}} + c'_{\mathrm{m}} ||q||^2,
\end{equation}
where $c_{\mathrm{m}}, c'_{\mathrm{m}}$ are two positive scalars. If the elongation is considered negligible, then $c'_{\mathrm{m}} = 0$.
Coriolis and centrifugal forces $C(q,\dot{q})\dot{q} \in \mathbb{R}^{n}$ can be evaluated using the standard mathematical machinery (e.g. Christoffel symbols).
Elastic $K(q) \in \mathbb{R}^{n}$ and gravitational $G(q) \in \mathbb{R}^{n}$ actions are defined as 
\begin{equation}\label{eq:potentials}
K(q) = \frac{\partial U_{\mathrm{K}}}{\partial q}, \quad G(q) = \frac{\partial U_{\mathrm{G}}}{\partial q},
\end{equation}
where the scalar-valued functions $U_{\mathrm{K}}$ and $U_{\mathrm{G}}$ are the associated potential energies. They are obtained as integration along the spatial coordinate of the energetic contributions of each infinitesimal elements.
The elastic force field is always positive definite, and thus the stiffness matrix
\begin{equation}\label{eq:K_pd}
    \frac{\partial K(q)}{\partial q} \succ 0.
\end{equation}
In some corner cases (e.g., floating base) this matrix may be semi-positive definite instead.
%
%
Gravitational forces are bounded as follows
\begin{equation}
||G(q)|| \leq c_{\mathrm{g}} + c'_{\mathrm{g}} ||q||,
\end{equation}
where $c_{\mathrm{g}}, c'_{\mathrm{g}}$ are two positive scalars, with the latter being equal to $0$ if no elongation is present.
The friction losses are usually modeled as a possibly nonlinear damping action $D(q)\dot{q}$, with $D \succ 0$. 
The input field $A(q) \in \mathbb{R}^{n \times m}$ is the transpose of the Jacobian mapping the $m \leq n$ actuation forces from their point of application to the configuration space.
Indeed, control actions are often not directly collocated on the states \cite{zaidi2021actuation}. Without loss of generality, we assume $A$ to be full rank columns. Some representative examples of actuation matrices are provided in Fig. \ref{fig:actuation}. 
\begin{figure}
    \centering
    \includegraphics[width = .8\textwidth]{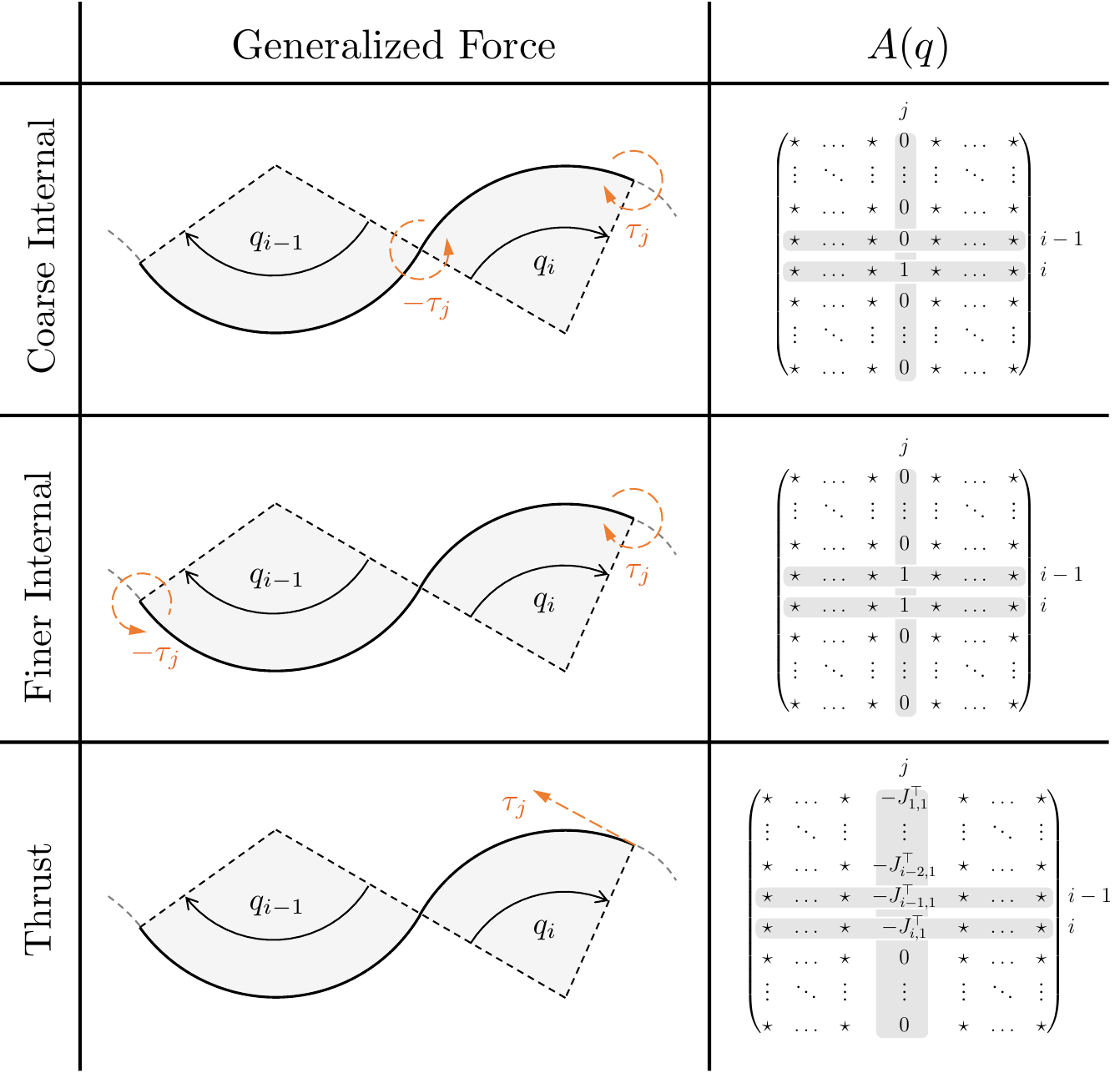}
    \caption{\small Examples of actuation patterns resulting in different elements of the matrix $A(q)$. The examples are focused on how a generalized force $\tau_j$ that is applied to the $i\--$th segment is mapped in the configuration space. Constant curvature discretization is used for illustrative purposes. The first two rows of the table report the case of one actuated segment being represented with one or two constant curvature segments (coarse and finer discretization respectively). The actuator produces an internal pair of opposite torques, as in figure. This is the case of a pair of pressurized chambers, or of a tendon driven system with motors placed at the base of the segment. In the third row of the table, $\tau_j$ is a force applied tangentially to the tip of the $i\--$th segment. Here, $J_{k,1}$ is the $k\--$th element of the Jacobian mapping $\dot{q}$ in the first linear velocity of the tip of segment $i$.}
    \label{fig:actuation}
\end{figure}

\subsubsection{Piecewise Constant Strain Approximations}
This family of discretization methods works by assuming that the strain $\xi$ is piecewise constant in $s$, with discontinuities happening at fixed points along the rod, called nodes. The right hand side of Fig. \ref{fig:rod_models} shows an example of one such a model. The most straightforward implementation of this principle is planar Piecewise Constant Curvature (PCC) models. Here, all strains but one curvature are neglected. The curvature itself is assumed to be piecewise constant. The resulting shape is a sequence of arcs connected in such a way that $x$ is everywhere differentiable, as shown in Fig. \ref{fig:deformations}.  The vector $q \in \mathbb{R}^n$ collecting all the local curvatures (one per CC segment) is the finite-dimensional configuration of a PCC robot. Thus, a PCC robot has as many degrees of freedom $n$ as the number of considered segments $n_{\mathrm{S}}$.
Soft robots under PCC approximation can be seen as a direct extension of serial manipulators with revolute joints to the continuum domain. Instead of being localized to one point (the joint), the change in angle is here homogeneously distributed along the segment.
%
Note indeed that the curvature is equivalent to the angle subtended by the CC arc - also called bending angle - since it is defined with respect to a normalized arc length.
The kinematics and dynamics of a single constant curvature segment are analyzed in sidebar \ref{sb:CC_derivation}\ref{?}.

The kinematic description \eqref{eq:kinematics} of PCC robots has been intensively used for more than a decade, as proven by the seminal review paper \cite{webster2010design} published in 2010.
The simplicity of their kinematics has been indubitably a major role in fostering this success, together with their effectiveness in describing real systems. Note indeed that a piecewise constant curvature is the exact steady state solution of the infinite dimensional model when only pure torques are applied along the structure.
When moving to 3D space, the geometrical characterization of a PCC robot as a sequence of arcs remains unvaried, but this time the plane of bending can change. This effectively introduces two degrees of freedom per segment, i.e. $n = 2 n_{\mathrm{S}}$.
The usual way in which this motion is represented in the literature is by including into $q$ the orientation of the $n$ relative orientations of the planes of bending. Although intuitive, this representation introduces singularities and discontinuities \cite{jones2007limiting} which can be avoided with alternative parametrizations \cite{chawla2018comparison,felt2019inductance,allen2020closed,della2020improved}. Occasionally piecewise constant elongation is also considered, with discontinuity points coincident to the curvature ones, thus leading to $n = 2 n_{\mathrm{S}}$ for the planar case and $n = 3 n_{\mathrm{S}}$ for the 3D one.

As discussed for the general case, the dynamic model \eqref{eq:dynamics} of a PCC robot is obtained combining \eqref{eq:kinematics} with the physical characteristics of the system through standard Lagrangian formalism \cite{godage2016dynamics}. 
Yet, when multiple segments are considered it is unpractical to derive closed form expressions of $M$ and $G$. Approximations of the mass distribution are thus often imposed to simplify the model derivations. Models having the mass of each segment lumped into a single along the rod are discussed in \cite{falkenhahn2015dynamic,wang2016three}. Alternatively a lumped mass and inertia can be placed at the center of mass \cite{godage2015accurate,godage2019center}, therefore neglecting only the change in rotational inertia. Under this hypothesis, the soft robot can be represented through an augmented rigid robot model, therefore enabling the use of standard tools for calculating the expression of dynamic forces \cite{della2020model,katzschmann2019dynamic,trumic2020adaptive}.

%
A simpler alternative to PCC which can be used to model robots bending and elongating are rigid link approximations \cite{kang2012dynamic,roesthuis2016steering}. The rod is here approximated through a sequence of links connected with standard independent joints. This is equivalent to considering a $\xi$ which is null everywhere, except for a finite set of points where it assumes the value of a Dirac's delta.
 Lumped springs are added in parallel to each joint to describe the robot's impedance. The resulting structure is a standard rigid robot with parallel elasticity, and it is therefore described by a set of ODE as \eqref{eq:dynamics}.
Kinematic model of parallel soft manipulators based on these strategies are discussed in  \cite{altuzarra2019position,nuelle2020modeling,chen2021kinetostatics}.

PCC models can be extended further by adding also piecewise constant shear deformation and twist. 
This is done by the piecewise constant strain models proposed in \cite{renda2018discrete,grazioso2019geometrically}, which define procedures for extracting models in the form \eqref{eq:dynamics}, where $q \in \mathbb{R}^{6n_{\mathrm{s}}}$. The inertia matrix $M(q)$ implements a full dynamic coupling within the six strains. On the contrary, the elastic part of the potential forces $K(q)$ is usually either diagonal or block diagonal. 
Thanks to their capability of describing complex strain conditions naturally arising in closed kinematic chains, these models can be used to describe soft parallel structures \cite{armanini2021discrete}.

\subsubsection{How Fine Should the Discretization Be?}
For a given soft robot to be modeled under a piecewise constant strain approximation, the number of segments $n_{\mathrm{s}}$ to be considered is up to the designer of the model to decide as a result of application-specific considerations. 
First, the mechanics of the robot imposes a lower bound in case the robot is obtained as a sequence of actuated modules, since considering less than one constant strain segment for each actuated one would generate incoherent behaviors (e.g. actuating one segment would result in a motion in the following one). Also, it is in general inconvenient to have constant strain segments shared between more than one actuated segment.
If the robot is used for simulation, then a trade-off between accuracy and simplicity must be established. For $n_{\mathrm{s}} \rightarrow \infty$ the model converges to the exact continuum representation, but the computational cost for the simulation will increase as $O(n_{\mathrm{s}}^2)$ at the best \cite{featherstone2014rigid}. 
In case the model is used for control design, then the amount and the location of the segments may change the structural properties of \eqref{eq:dynamics}.
It is for example common to place segments in such a way that the resulting model is fully actuated, i.e. $A(q)$ is square and full rank.
For planar PCC models, if $\tau$ are torques applied at the tip of each CC segment then $A(q) = I$ - therefore further strengthening the parallelism with standard serial rigid robots. Most of existing actuation technologies will satisfy this hypothesis (e.g. tendons, fluids).

\subsubsection{Functional Parametrizations}
Instead of discretizing along the arclength, the reduction of dimensionality can happen by projecting onto a low dimensional functional subspace as follows
\begin{equation}\label{eq:functional_xi}
    \xi(s,t) = \sum_{i = 1}^{n} \pi_i(s) q_i(t),
\end{equation}
where $\{\pi_1(s),\dots,\pi_n(s)\}$ is a base of the subspace, 
and the weights $q_i(t)$ can be taken as the configuration of the robot. 
One simple way of selecting $\pi_i(s)$ is to truncate an infinite dimensional basis of a regular-enough functional space. In this way, it is ensured/guaranteed that the approximation converges to the exact model for $n \rightarrow \infty$. It is also convenient to include constant functions, so that the model is a proper extension of the constant curvature or constant strain ones.
For example, in the case of inextensible planar soft robots (i.e. $\xi$ contains only the curvature), polynomials $\pi_i(s) = s^{i-1}$ are a basis that satisfies both conditions.
This modeling technique is widely used in flexible link robots \cite{de2016robots} to represent link vibrations. In this case, the base functions $\pi_i(s)$ are selected as the $n$ slowest modes of the Euler Bernoulli beam modeling the link \cite{book1984recursive,de1993inversion}.
Functional expansions are also used in continuum mechanics to approximate the equilibria of rods and beams \cite{echter2010numerical,greco2013b}.
The application to rod-like structures originated from the field of hyper-redundant robots \cite{chirikjian1994modal}. Here, the rod serves as a continuum approximation of the discrete system \cite{mochiyama1999shape}. 
The dynamics of polynomial curvature soft robots is discussed in \cite{della2019control,della2020soft}.
Alternatively, it is the position part of Cartesian configuration $x$ that can be projected onto a polynomial space \cite{sadati2017control,singh2018modeling} (polynomial shape robot), or on a base derived as the truncated Taylor expansion of the forward kinematics \cite{godage2011shape,godage2011novel}.
This strategy is also investigated in \cite{navarro2017fourier,qi2020adaptive,palli2020model} when modeling deformable objects. 

No matter the choice of $\pi_i$, the robot's dynamics can still be formulated as \eqref{eq:dynamics}. However, a full dynamic and potential coupling among all the degrees of freedom exists in general. Finally, the piecewise and functional strategies can be combined into piecewise variable strain models \cite{boyer2020dynamics}. In this case, the transformation matrices in the right hand side of Fig. \ref{fig:rod_models} are parametrized using a functional approach.

\subsection{Finite Element Models}

\begin{figure}[t!]
	\centering
	\includegraphics[width=.475\textwidth]{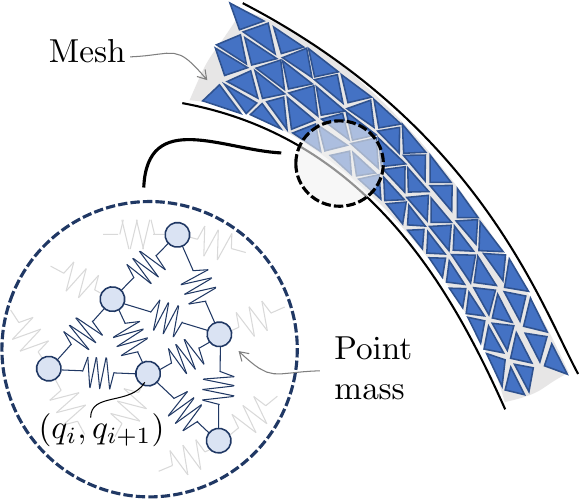}
	\caption{\small The dynamics of the planar portion of a soft robot is approximated through FEM as a network interconnection of points mass and springs. A mesh defines the topology of this interconnection. The configuration $q$ of the robot is defined as the collection of the Cartesian positions of all masses. \label{fig:FEM_models}}
\end{figure}

In Finite Element Models (also Methods, FEM) of deformable solids \cite{rust2015non}, the geometric shape is described by identified a mesh, which is a set of nodes together with the information of which are their neighbors (Fig. \ref{fig:FEM_models}). If the position of nodes is known, an approximation of the entire volume results from interpolation.  FEM is, therefore, the preferred solution whenever the changes in the three-dimensional structure of the robot are not negligible compared to its virtual backbone. In their general definition, finite element methods are formulated as ways of approximating solutions of PDEs.
As such, FEM can be used to discretize rod models of soft robots \cite{grazioso2019geometrically}. Furthermore, this framework naturally extends to modeling systems encompassing multiple continuous behaviors - e.g. magnetic, thermal, fluid \cite{zimmerman2006multiphysics}.

The discussed space interpolation naturally leads to a kinematic description in the form of \eqref{eq:kinematics}. The configuration $q \in \mathbb{R}^{n}$ is the collection of the nodes location in the space. So, $n$ is in general three times the number of nodes. Since the full volume is explicitly considered, the forward kinematics $h(s,q)$ is to be parametrized with $s\in\mathbb{R}^3$, rather than a scalar.  
Similarly to rod models, FEM have the advantageous property of converging to the exact model when $n$ tends to infinity. Using several thousand nodes, in general, produces a very accurate model at the cost of a quite large configuration space.
Note however that measuring or observing the whole state of a FEM model is often not needed when implementing closed loop controls. 
Dynamic equations in the form \eqref{eq:dynamics} result from the application of Lagrangian machinery to \eqref{eq:kinematics}, in a similar fashion as for the rod case.
Discretization in strain space $\xi$ usually results in a linear $K(q)$ and constant $D(q)$, at the cost of a configuration dependent inertia $M(q)$. 
On the contrary, in FEM analysis the following simplifications are introduced whenever the mass is assumed concentrated to the nodes: ${\partial M(q)}/{\partial q} = 0$, $C(q,\dot{q}) = 0$, and ${\partial G(q)}/{\partial q} = 0$.
Thus, the multi-body dynamic part of \eqref{eq:dynamics} simplifies into $M \ddot{q} + G$. Furthermore, $M$ is diagonal if $q$ represents the nodes' configuration in the space, which implies that there is no dynamic coupling.
On the downside $K(q)$ is usually nonlinear, and $D(q)$ and $A(q)$ are rarely constant.

\subsubsection{Model order reduction}
The high dimensional configuration space of a FEM model can be compressed by selecting a set of $n_{\mathrm{r}} << n$ directions of interest \cite{qu2004model}. In practice, $n_{\mathrm{r}}$ is usual in the order of a few dozens, and thus $n/n_{\mathrm{r}} \gtrsim 10^2$. The reduced order configuration is defined as $q_{\mathrm{r}} = \Phi q$, where for simplicity it is assumed that the configuration $q$ is defined such that $q = 0$ is the system equilibrium for $\tau = 0$.
This projection is conceptually similar to the functional parametrizations discussed above for rod dynamics. 
The resulting dynamics in $q_{\mathrm{r}}$ space is
\begin{equation}\label{eq:dynamics_reduced_FEM}
\left(\Phi^{\top} M \Phi\right) \ddot{q}_{\mathrm{r}} + \Phi^{\top} G + \left(\Phi^{\top} D(\Phi q_{\mathrm{r}}) \Phi\right) \dot{q}_{\mathrm{r}} + \Phi^{\top} K(\Phi q_{\mathrm{r}}) = \Phi^{\top} A(\Phi q_{\mathrm{r}}) \tau,
\end{equation}
which is again in the general form \eqref{eq:dynamics}. Note that $\Phi^{\top} M \Phi$ is diagonal if $M$ is diagonal and if $\Phi$ is orthogonal.
The matrix $\Phi$ should be selected in such a way that the solution of \eqref{eq:dynamics_reduced_FEM} is representative of the evolution of \eqref{eq:dynamics} for the full FEM model, as soon as the initial condition satisfies $(q_{\mathrm{r}}(0), \dot{q}_{\mathrm{r}}(0)) = (\Phi q(0), \Phi \dot{q}(0))$.
Modal analysis is a well known tool in FEM theory to achieve this goal \cite[Ch. 12]{qu2004model}. The linear modes of the linearized system about the equilibrium configuration are calculated. The eigenvectors associated to the smallest eigenvalues (slowest modes) are used to build $\Phi$. The value of $n_{\mathrm{R}}$ can be defined according to the frequency range of vibrations that the designer is interested in capturing. This procedure implicitly operates a regularization of the FEM model, getting rid of many numerical vibrations happening at high frequency, which can be assimilated to numerical noise. 
Alternatively, the columns of $\Phi$ can be evaluated from the singular value decomposition of a dataset of representative evolutions of the system \cite{goury2018fast}. 
These analyses can be extended to nonlinear deformations \cite{kerschen2009nonlinear,sin2013vega,chenevier2018reduced}.  
The extension to the modeling of soft robot with self-contact forces is discussed in \cite{goury2021real}. 
One reason for $n$ to reach high values - up to hundreds of thousands - is that the robot geometry has many details (thinned areas, holes, small grooves). 
This issue can be tamed through several methods as for example X-FEM \cite{oliver2006comparative,yazid2009state}, which however does not radically reduce the size of the models.
Condensation is another widely used approach. It operates model reduction by partitioning the configuration space in loaded - directly actuated - and unloaded variables. The loading conditions are modeled as holonomic constraints and solved with Lagrangian multipliers \cite{paz1984dynamic,qu2004model,duriez2013control}.
Condensation has also been used to connect FEM and rod models in \cite{adagolodjo2021coupling}.
A recent review paper on model reduction techniques is also available \cite{sadati2021reduced}.

\begin{figure}[th!]
	\centering
	\includegraphics[width=.4\textwidth]{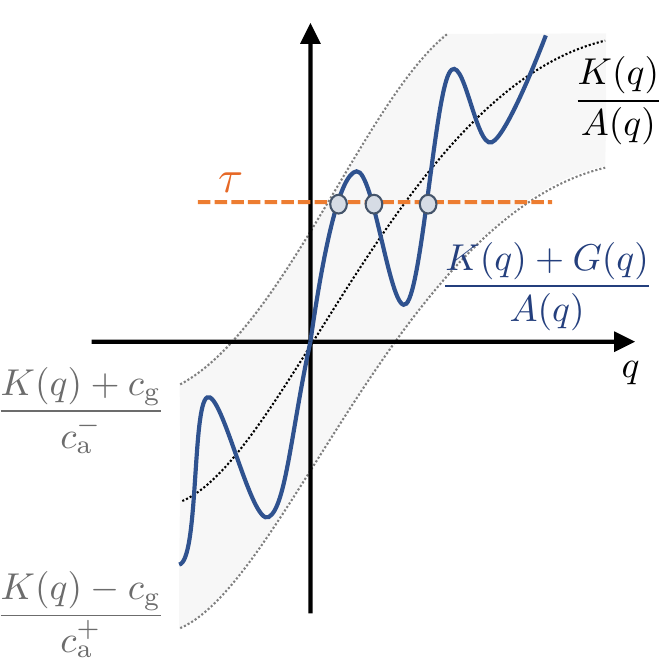}
	\caption{\small  In the scalar case, equation \eqref{eq:equilibrium} always has at least one solution for unbounded $K$ and bounded $A$ and $G$. This figure reports a pictorial representation of the reasoning behind the proof of this statement. \label{fig:equilibrium}}
\end{figure}

\subsection{Existence of Equilibria}

The equilibrium configurations of \eqref{eq:dynamics} associated to a control input $\bar{\tau}$ are all the $\bar{q} \in \mathbb{R}^n$ such that 
\begin{equation}\label{eq:equilibrium}
	K(\bar{q}) + G(\bar{q}) = A(\bar{q}) \bar{\tau}.
\end{equation}
Due to the many nonlinearities involved, in general the solutions of \eqref{eq:equilibrium}  cannot be expressed in closed form. 
An exception is when $K + G$ is monotone (e.g., high enough stiffness) 
%
%
and the input field $A$ is configuration independent. In this case the solution of \eqref{eq:equilibrium} is
\begin{equation}\label{eq:equilibium_nog}
    \bar{q} = (K + G)^{-1} (A \, \bar{\tau}).
\end{equation}
So, a single equilibrium always exists for any choice of actuation.
If $K + G$ is not monotone or if the actuation field $A$ is configuration dependent, the existence of at least one equilibrium for any given $\bar{\tau}$ is to be expected if $K(q)$ is radially unbounded (i.e. stiffness not vanishing). Several solutions to this equation may in general exist.
Consider for example the case of $n=1$ and $c_{\mathrm{a}}^- < A < c_{\mathrm{a}}^+$. Thus, if $K(q)$ is radially unbounded and $G(q)$ is limited, then also $K/A$ and $G/A$ are. Thus, $(K + G)/A$ is radially unbounded, even if in general not monotonic. It is also a continuous function.
As a consequence, there is always at least one equilibrium configuration, that is a configuration $\bar{q}$ that verifies $(K(\bar q) + G(\bar q))/A(\bar q) = \tau$. This sketch of proof is visually represented by Fig. \ref{fig:equilibrium}.
The existence of at least one equilibrium for any constant actuation is in sharp contrast with classic rigid robots, for which a constant actuation can never result in an equilibrium configuration unless gravity is involved.

\subsection{Actuators dynamics}
Soft robots are actuated through a wide variety of strategies. Yet, relatively few attention has been devoted so far to incorporating the dynamics of actuators in dynamic models used for control.
Nonetheless, we can still provide a model based on similar formulations from classic robotics \cite{spong1987modeling,della2021flexible}, which holds for actuation strategies where the main function components of the actuators are themselves mechanical (e.g. tendons actuated through electric motors, fluids pressurized through pistons)
\begin{align}
M(q)\ddot{q} + C(q,\dot{q})\dot{q} + D(q)\dot{q} + K(q) + G(q) + \frac{\partial U_{\mathrm{c}}}{\partial q}(q,\eta) &= 0, \label{eq:dynamics_r_act}\\
B(\eta)\ddot{\eta} + H(\eta,\dot{\eta})\dot{\eta} + 
\frac{\partial U_{\mathrm{c}}}{\partial \eta}(q,\eta) &= \tau,\label{eq:dynamics_a_act}
\end{align}
where we not include dissipation in the actuation system for simplicity. The configuration of all the actuators is collected in $\eta \in \mathbb{R}^{m}$
, and $B,H \in \mathbb{R}^{m\times m}$ are the associated inertia and Coriolis matrices. The former is usually diagonal and configuration independent, and in turn $H = 0$. This is however not always the case, an exception being magnetically actuated soft robots with magnets moved by a rigid robot \cite{pittiglio2019magnetic}.

The coupling between the dynamics \eqref{eq:dynamics_r_act} and \eqref{eq:dynamics_a_act} is purely mediated by the potential field $U_{\mathrm{c}}$, which models elasticity of tendons, molecular interactions in compressible fluids, or electro-magnetic fields, just to cite a few.
In case the dynamics of $\eta$ is fast compared to $q$, as well as robustly globally stable, then \eqref{eq:dynamics_a_act} can be approximated with its steady state behavior $\eta \simeq \bar\eta(q,\tau)$.  In this case ${\partial U_{\mathrm{c}}} (q,\bar\eta(q,\tau)) /{\partial q}$ serves as a generalization of the input field $A(q)\tau$ appearing in \eqref{eq:dynamics}. 
Alternatively, singular perturbation theory can be used to separate the fast actuator dynamics from the slow soft robot one, without applying quasi-static approximations \cite{kokotovic1976singular}.

\subsection{Simulators}
A bottleneck to entering into the field of soft robots control has been the need to implement the simulator of the soft robot. This is especially troublesome when considering that the models used for simulation are typically way more sophisticated than the ones used for control. Luckily, there are now several open-source solutions available: SOFA \cite{allard2007sofa,coevoet2017software} and ChainQueen \cite{hu2019chainqueen} use velumetric FEM techniques, while Elastica \cite{naughton2021elastica}, TMTDyn \cite{sadati2019tmtdyn}, SimSOFT \cite{grazioso2019geometrically}, and SoRoSim \cite{mathew2021sorosim} implement discretizations of rod modes. More details on simulators for soft robots can be found in \cite[Sec. VII]{collins2021review}.
Still, selecting the right model among all the available ones is a task with no clear solution. Experimental comparisons as the ones provided in \cite{chikhaoui2019comparison,rao2021model} can be a useful tool in this context.

\section{Shape Control in the Fully Actuated Approximation}
The primary task of control architectures in classic robotics is to accurately manage the posture of the robot - i.e. state space control. In the case of soft robots, this translates into devising strategies to control the whole shape of the system, that is controlling $q$. Depending on the model used for control design, this task may translate into different goals - as for example curvature, strain, or volume control - which however share a same set of characteristics.  The importance of carefully selecting the model used for control design becomes therefore apparent. Indeed, applying a same control solution with different models will in general produce substantially different closed loop behaviors - both in terms of transient and steady state.
We start in this sections with robots that can be effectively modeled as fully actuated - i.e., $m = n$ and without loss of generality $A(q) = I$. We will see how this approximation allows to already acquire important insight on the behavior of soft robots.

\begin{figure}[th!]
	\centering
	\includegraphics[height=.33\textwidth]{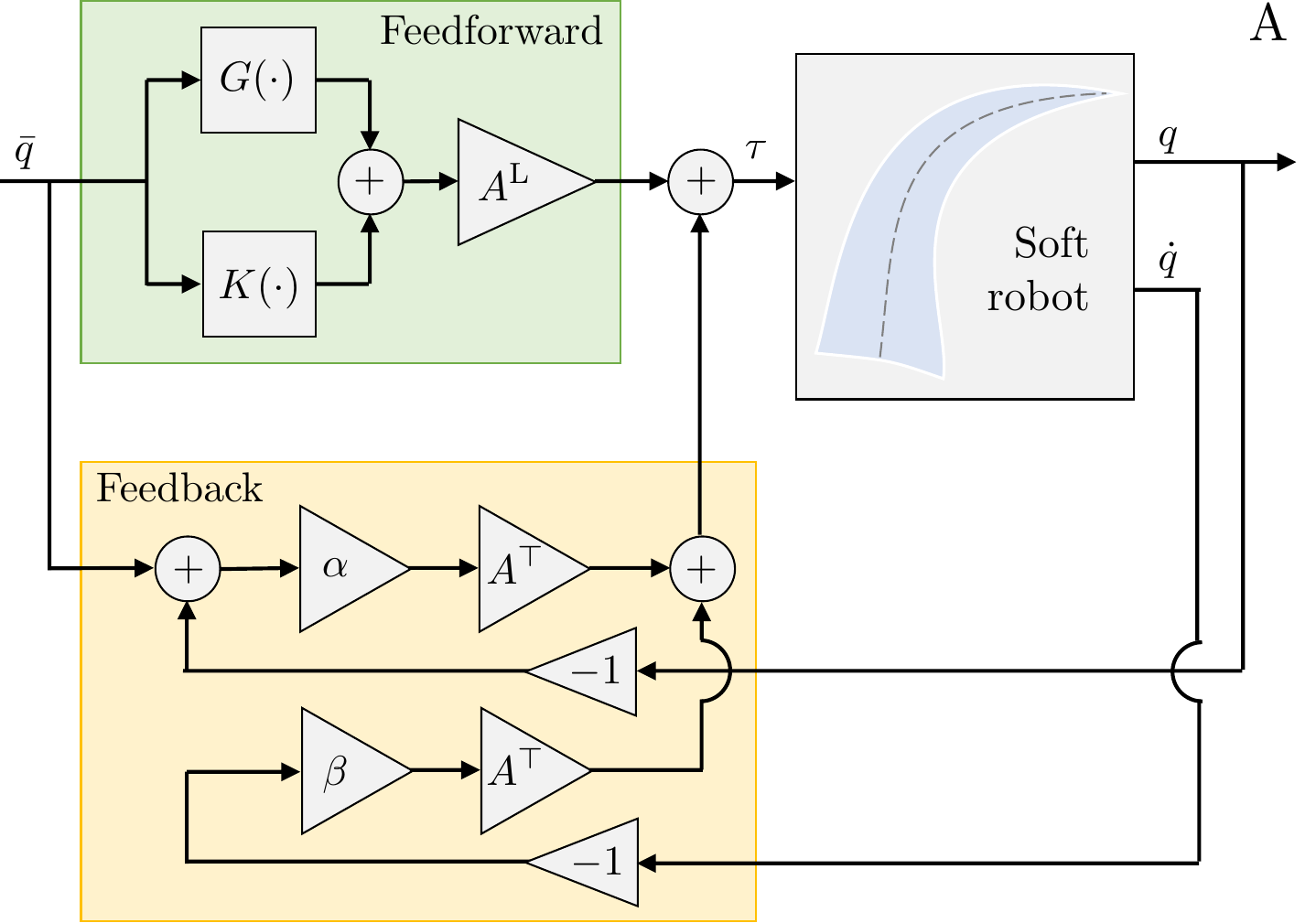}
	\hspace{0.01\columnwidth}
	\includegraphics[height=.33\textwidth]{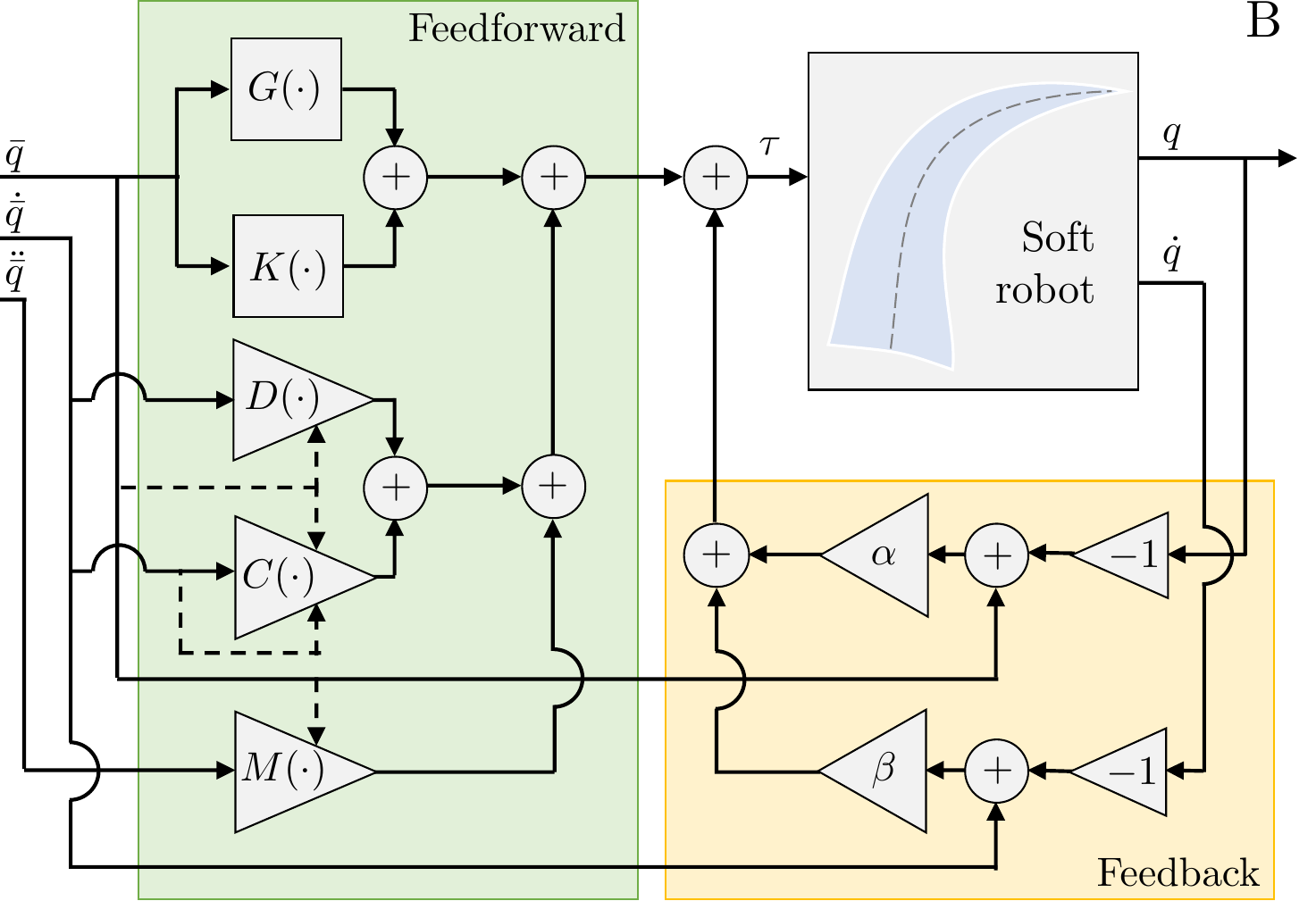}
	\caption{\small  Block schemes of PD-like controllers for shape regulation and tracking. Panel A depicts the regulator \eqref{eq:PD_ua}, which corresponds to \eqref{eq:PD_setpoint} when $A = I$, to \eqref{eq:ff_under} when $\alpha,\beta$ are null, and the pure feedforward action \eqref{eq:feedforward_dynamic} when both conditions are fulfilled. Panel B shows the high gain PD controller with feedforward compensation \eqref{eq:PD_trajectory} for trajectory tracking under fully actuated approximations. If the dashed lines are connected to $q$ and $\dot{q}$ instead of the reference, the controller directly extends the classic PD+. Feedforward and feedback components are highlighted in figure.
	\label{fig:regulation_tracking_block}}
\end{figure}


%
%
%
\subsection{Posture regulation}
%
Posture regulation is defined as follows: given a desired constant configuration $\Bar{q} \in \mathbb{R}^{n}$ find a control action $\tau \in \mathbb{R}^{m}$ such that the configuration of the soft robot $q \in \mathbb{R}^{n}$ eventually converges to the desired one, i.e.,
\begin{equation}
    \lim_{t \rightarrow \infty} q(t) = \Bar{q}.
\end{equation}

It has been already discussed in the previous section that an equilibrium is always associated to any constant control input -  as exemplified by \eqref{eq:equilibium_nog}. We show here that this equilibrium is also asymptotically stable under opportune conditions on the mechanical impedance of the robot.
Consider the following purely feedforward controller (Fig. \ref{fig:regulation_tracking_block}) 
\begin{equation}\label{eq:feedforward_dynamic}
    \tau(\Bar{q}) = K(\Bar{q}) + G(\Bar{q}),
\end{equation}
where $K$ and $G$ are the elastic and gravitational fields with potentials $U_{\mathrm{K}}$ and $U_{\mathrm{G}}$ respectively, as defined in \eqref{eq:potentials}.
Substituting \eqref{eq:feedforward_dynamic} in \eqref{eq:dynamics} and rearranging terms yield
\begin{equation}\label{eq:dynamics_regulation_ff}
    M(q)\ddot{q} + C(q,\dot{q})\dot{q} = \underbrace{(K(\bar{q}) + G(\bar{q})) - (K(q) + G(q))}_{\text{Physical P-loop}} + \underbrace{D(q)(-\dot{q})}_{\text{Physical D-loop}},
\end{equation}
where we can recognize the same mathematical structure of a classic robot (left hand side) controlled through a nonlinear PD regulator (right hand side). 
Note indeed that the reference is constant, and thus $\dot{\bar{q}} = 0$.
It is worth stressing that the physical system is excited with a simple feedforward at this stage, which only behaves as a PD regulator when combined with part of the robot's dynamics.

The control community has devoted much attention to (nonlinear) PD controllers \cite{aastrom2001future}, 
which has produced a thriving literature which soft roboticist can borrow from \cite{kelly1996class,kelly1997pd,kelly1998global}, by relying on \eqref{eq:dynamics_regulation_ff}, as for example done in the following theorem.

\begin{Theorem}\label{th:ff}
The state $(\bar{q},0) \in \mathbb{R}^{2n}$ is an  asymptotically stable equilibrium of system \eqref{eq:dynamics} subject to the constant control action \eqref{eq:feedforward_dynamic} if an open neighbourhood $\mathcal{N}(\bar{q}) \subseteq \mathbb{R}^n$ of $\bar{q}$ exists such that $\forall q \in \mathcal{N}(\bar{q})/\{\bar{q}\}$,  
\begin{equation}\label{eq:potential_high}
    \left(U_{\mathrm{G}}(q) + U_{\mathrm{K}}(q)\right)  > \left(U_{\mathrm{G}}(\bar{q}) +  U_{\mathrm{K}}(\bar{q})\right) + \left.\left(\frac{\partial}{\partial q} \left(U_{\mathrm{G}}(q) + U_{\mathrm{K}}(q)\right)\right)\right\rvert_{q = \bar{q}}^{\top}\!\!(q - \bar{q}),
\end{equation}
and
\begin{equation}\label{eq:isolated}
    G(q) + K(q) \neq G(\bar{q}) + K(\bar{q}).
\end{equation}
These two conditions also imply that $\mathcal{N}(\bar{q})$ is fully included within the region of asymptotic stability of $\bar{q}$.
\end{Theorem}
\begin{proof}
Consider as Lyapunov candidate the following generalization of the energy of the robot
\begin{equation}\label{eq:lapunov_ff}
    V(q,\dot{q}) = \underbrace{\frac{1}{2}\dot{q}^{\top} M(q) \dot{q}}_{\text{Kinetic energy}} + \underbrace{U_{\mathrm{G}}(q) - U_{\mathrm{G}}(\bar{q}) + U_{\mathrm{K}}(q) - U_{\mathrm{K}}(\bar{q})}_{\text{Centered potential energy}} + \underbrace{(G(\bar{q}) + K(\bar{q}))^{\top}(\bar{q} - q)}_{\text{Correction term}}.
\end{equation}
The kinetic energy is always strictly positive definite in $\dot{q}$ since $M \succ 0$. Thus, a necessary and sufficient condition for $V$ to be positive definite in $(q,\dot{q})$ is that $V - \dot{q}^{\top} M(q) \dot{q}/2$ is positive definite in $q$, which is equivalent to \eqref{eq:potential_high}.
The next step is to study the sign of the time derivative of \eqref{eq:lapunov_ff}, which is
\begin{equation}\label{eq:V_dot}
    \begin{split}
        \dot{V}(q,\dot{q}) &= \dot{q}^{\top} M(q) \ddot{q} + \frac{1}{2}\dot{q}^{\top} \dot{M}(q) \dot{q} - \dot{q}^{\top} \left((K(\bar{q}) + G(\bar{q})) - (K(q) + G(q)) \right)\\
        &= -\dot{q}^{\top} C(q,\dot{q}) \dot{q} + \dot{q}^{\top}D(q)(-\dot{q}) + \frac{1}{2}\dot{q}^{\top} \dot{M}(q) \dot{q} \\
        &= - \dot{q}^{\top}D(q)\dot{q},
    \end{split}
\end{equation}
where the first step exploits \eqref{eq:dynamics_regulation_ff} to express $M \ddot{q}$, and the second the passivity of the system $\dot{q}^{\top} (\dot{M}(q) - 2 C(q,\dot{q})) \dot{q} = 0$.
Eq. \eqref{eq:V_dot} is only semi-positive definite despite $D(q)$ being a strictly positive matrix, since $\dot{V}(q,0) = 0$ for all $q$.
Thanks to LaSalle's principle, the system converges to the set of $(q,0)$ such that $\ddot{q} = 0$. To conclude the proof it is therefore sufficient to show that $\bar{q}$ is the only configuration in $\mathcal{N}(\bar{q})$ such that $\ddot{q} \neq 0$ for $\dot{q} = 0$, i.e.,
\begin{equation}
	G(q) + K(q) \neq \bar{\tau},
\end{equation}
which thanks to \eqref{eq:feedforward_dynamic} is equivalent to the hypothesis \eqref{eq:isolated}, thus yielding thesis.
\end{proof}

Eq. \eqref{eq:potential_high} is a convexity condition on the total potential energy $U_{\mathrm{G}}(q)+ U_{\mathrm{K}}(q)$. As such, it can be locally checked by looking at the sign of the Hessian matrix. This results in the condition
\begin{equation}\label{eq:stablity_condition_ff}
    \left. \left(\frac{\partial K(q)}{\partial q} + \frac{\partial G(q)}{\partial q}\right) \right\rvert_{q = \bar{q}} \succ 0,
\end{equation}
which is to say that the force field linearized at the desired equilibrium is attractive. In turn, this also implies that the potential force field is locally not constant, therefore implying also that hypothesis \eqref{eq:isolated} is true at least in an infinitesimal neighborhood of $\bar{q}$. Additionally, \eqref{eq:stablity_condition_ff} becomes a necessary condition for \eqref{eq:potential_high} when $\succeq$ is used instead of $\succ$.
The two terms in \eqref{eq:stablity_condition_ff} are the stiffness matrices associated to elastic and potential fields. While the first is always positive definite - see  \eqref{eq:K_pd} - the second is in general not definite in sign.
Gravity may serve either as a destabilizing ($\partial G(q) / \partial q \preceq 0$) or as stabilizing force ($\partial G(q) / \partial q \succeq 0$).  
For the CC segment described in sidebar \ref{sb:CC_derivation}\ref{?}, these two conditions corresponds to the robot pointing upwards ($\phi = \pi$) or downwards ($\phi = 0$) when in straight configuration ($q = 0$) respectively.

Thus, as already pointed out for the analysis of equilibria, the presence of an elastic field makes the control problem simpler to solve compared to the standard rigid case. This can be regarded as an instance of the so-called self stabilization property of soft robots, which has been recognized by several works in the literature \cite{thuruthel2018stable,bosi2016asymptotic}. 
However, even if a feedforward action has proven to be sufficient for stiff enough systems, it is still interesting to consider what happens when a further feedback loop is introduced.
This may serve several purposes, as for example enlarge the basin of attraction, shape the transient, and reject disturbances.
Further following along with the analogy with nonlinear PDs, \eqref{eq:feedforward_dynamic} can be extended as follows for the fully actuated case (Fig. \ref{fig:regulation_tracking_block})
\begin{equation}\label{eq:PD_setpoint}
    \tau(\Bar{q},q,\dot{q}) = \underbrace{K(\Bar{q}) + G(\Bar{q})}_{\text{Feedforward}} + \underbrace{\alpha (\Bar{q} - q) - \beta \dot{q}}_{\text{PD}}. 
\end{equation}
Here, $\alpha, \beta \in \mathbb{R}^{n \times n}$ are two gain matrices weighting the proportional and derivative actions respectively.
\begin{Corollary}\label{cl:fully_feedback}
    The state $(\bar{q},0) \in \mathbb{R}^{2n}$ is an asymptotically stable equilibrium of the closed loop \eqref{eq:dynamics}-\eqref{eq:PD_setpoint} if $D(q) \succ - \beta$, and an open neighbourhood $\mathcal{N}(\bar{q}) \subseteq \mathbb{R}^n$ of $\bar{q}$ exists such that $\forall q \in \mathcal{N}(\bar{q})/\{\bar{q}\}$,  
    \begin{equation}
    \left(U_{\mathrm{G}}(q) + U_{\mathrm{K}}(q)\right)  + \frac{1}{2} (q - \bar{q})^{\top} \alpha (q - \bar{q}) > \left(U_{\mathrm{G}}(\bar{q}) +  U_{\mathrm{K}}(\bar{q})\right) + \left.\left(\frac{\partial}{\partial q} \left(U_{\mathrm{G}}(q) + U_{\mathrm{K}}(q)\right)\right)\right\rvert_{q = \bar{q}}^{\top}\!\!(q - \bar{q}),
    \end{equation}
    and
    \begin{equation}
    G(q) + K(q) + \alpha (\Bar{q} - q) \neq G(\bar{q}) + K(\bar{q}).
    \end{equation}
    These two conditions also imply that $\mathcal{N}(\bar{q})$ is fully included within the region of asymptotic stability of $\bar{q}$.
\end{Corollary}
%
%
\begin{proof}
	The closed loop dynamics is $M(q)\ddot{q} + C(q,\dot{q})\dot{q} = \left(K(\bar{q}) - K(q)\right) + (G(\bar{q}) - G(q)) + \alpha (\bar{q} - q) - (D(q) + \beta)\dot{q}$.
    The previously discussed proof generalizes to this case by adding $(\bar q - q)^{\top} \alpha (\bar q - q)/2$ to \eqref{eq:lapunov_ff}.
    The time derivative of this new Lyapunov candidate is $\dot{V} = - \dot{q}^{\top}\left( D(q) + \beta \right)\dot{q}$, which is semi-negative definite if $D(q) + \beta \succ 0$. So any $\beta \succeq 0$ implements a damping injection that does not destabilize the closed loop. The rest of the proof follows as in the feedforward case.
\end{proof}

The sufficient condition for local asymptotic stability is 
\begin{equation}\label{eq:stablity_condition_fb}
    \left. \left(\frac{\partial K(q)}{\partial q} + \frac{\partial G(q)}{\partial q}\right) \right\rvert_{q = \bar{q}} + \alpha \succ 0,
\end{equation}
which becomes necessary when only semi-positiveness is required.
Note that \eqref{eq:stablity_condition_fb} can always be fulfilled through a large enough proportional gain $\alpha$. Yet, large gains may result in a stiffening of the soft robot \cite{della2017controlling}, and in amplification of noise or excitation of neglected dynamics.  
%
Possibly nonlinear integral actions can also be added to \eqref{eq:PD_setpoint} for compensating steady state errors and achieve global stabilization, as discussed in \cite{kelly1998global}.

\subsection{Trajectory tracking}

In trajectory tracking the desired behavior is specified as an evolution of the full robot shape in time.
Consider a twice differentiable function of time $\Bar{q}: \mathbb{R} \rightarrow \mathbb{R}^{n}$, then the control goal is to find a control strategy $\tau$ such that
\begin{equation}
    \lim_{t \rightarrow \infty} (q(t),\dot{q}(t)) - (\bar{q}(t),\dot{\bar{q}}(t))= 0.
\end{equation}
Usually the reference is considered bounded in norm $||(\bar{q}(t),\dot{\bar{q}}(t))|| < c_{\mathrm{t}}$, for some positive $c_{\mathrm{t}}$.
In theory, under the fully actuated approximation $n = m$, \eqref{eq:dynamics} can be completely feedback linearized with a computed torque scheme. However, such a strategy would be hardly applicable on a real system. 
This section will focus on controllers achieving the trajectory tracking goal by relying minimally on direct model cancellations. For the sake of space, proof of convergence will not be provided. All of them can be obtained by adapting proofs from the nonlinear PD literature so to work for a system as \eqref{eq:dynamics_regulation_ff}, similarly as what it has been shown in Theorem \ref{th:ff}.

If the reference trajectory is slow varying (i.e. $||\dot{\Bar{q}}||$ small enough) then \eqref{eq:feedforward_dynamic} and \eqref{eq:PD_setpoint} can still be applied as they are, possibly with the inclusion of damping feedforward compensation terms - i.e., $D(\bar{q})\dot{\bar{q}}$ and $(D(\bar{q}) + \beta)\dot{\bar{q}}$ respectively. The state will not converge to $(\bar{q},\dot{\bar{q}})$ at steady state, but to a neighborhood of it 
\cite{kawamura1988local,cervantes2001pid}. Higher the gains and slower the reference, smaller is the neighborhood.

Explicit compensation of dynamic forces is needed for achieve null steady state error. Again, this can happen by largely relying on feedforward actions (Fig. \ref{fig:regulation_tracking_block})
\begin{equation}\label{eq:PD_trajectory}
    \tau(\Bar{q},\dot{\Bar{q}},\ddot{\Bar{q}},q,\dot{q}) =
    \underbrace{M(\bar{q})\ddot{\bar{q}} + C(\bar{q},\dot{\bar{q}})\dot{\bar{q}} + D({q})\dot{\bar{q}} + K(\Bar{q}) + G(\Bar{q})}_\text{Feedforward}
     + \underbrace{\alpha (\Bar{q} - q) + \beta (\dot{\bar{q}} - \dot{q})}_\text{PD}. 
\end{equation}
By adapting the results in \cite{kelly1994pd}, we can prove that \eqref{eq:PD_trajectory} leads to local exponential stabilization of the desired trajectory if $\alpha + \partial K/\partial q$ and $\beta + D$ are larger than two bounds which increase with the increase of $||\partial G/\partial q||$, $||\dot{\bar{q}}||$, and $||\ddot{\bar{q}}||$.
Purely feedforward dynamic controllers for soft robots are experimentally validated in \cite{marchese2016dynamics,falkenhahn2015model}.
Thus, if the reference is slowly-varying or the natural impedance is high enough, then the feedback gains $\alpha, \beta$ can be selected null and the controller is once more purely feedforward. Yet, for generic trajectories the discussed condition will be hardly verified by the physical properties and the extra feedback will be necessary.
As an alternative to high gains, Eq. \eqref{eq:PD_trajectory} can be further evolved into a partially nonlinear closed loop. This is done by evaluating $M$, $C$, and $G$ on the measured state $(q,\dot{q})$ rather than on the reference $(\bar{q},\dot{\bar{q}})$. This produces a closed loop which is equivalent to a rigid robot controlled through a PD+ controller \cite{paden1988globally}. In this case we can achieve perfect tracking even when $\alpha = 0$ and $\beta = 0$. This control strategy is discussed and experimentally validated on a soft robotic platform in \cite{della2020model}. A robust version of this controller which does not require direct measurement of $\dot{q}$ is proposed and tested in \cite{cao2021model}, which is close to the robust PD controller proposed in \cite{ouyang2014pd}.
Thanks to \eqref{eq:dynamics_regulation_ff}, latest advancements in PD control of mechanical systems \cite{su2009global,su2014finite,zhang2017pid} can be also adapted to further improve performances and robustness.

Several other works in the soft robotics literature deal with the trajectory tracking challenge by relying on fully actuated approximations. A computed torque controller built on a planar PCC model, and a sliding mode variation are experimentally tested in \cite{kapadia2014empirical}. A model based decentralized controller is proposed in \cite{doroudchi2021configuration}, and applied to a fully actuated discretization of the Cosserat model. Finally, \cite{george2020first} derives model based controllers under a first order approximation - i.e. when $||D(q)\dot{q}|| >> ||M(q)\ddot{q} + C(q,\dot{q})\dot{q}||$. This is for example the case of lightweight soft robots moving in viscous fluids \cite{milana2019artificial}. 

\section{Advanced Control Challenges: \\ Underactuation, Actuators Dynamics, and Task Execution}
%
%
This section discusses challenges that 
are largely unexplored, and still in need of general solutions and formulations. 

\subsection{Dealing with undearctuation in shape control}

Fully actuated approximations have proven to be effective in practice, despite being a clear over-simplification of th e control problem. 
By bringing under-actuation into the picture, the degrees of freedom not directly affected by the control action can be analyzed and potentially used in the design of the controller, towards solutions with improved performance and certifiable reliability.
Thus, consider a non-square actuation matrix $A(q)$ with $m < n$.
The first difficulty that arises is that the desired shape $\bar{q}$ may not be an attainable equilibrium of the system, i.e. $K(\Bar{q}) + G(\Bar{q}) \notin \mathrm{Span}(A(\bar{q}))$. In other terms, for a generic shape it will not exists a control action which makes it an equilibrium configuration. Similarly, in general it will not necessarily exist a control input evolution $\tau(t)$ such as a generic state $(\bar{q},\dot{\bar{q}})$ can be reached from any initial condition. Authors of \cite{zheng2019controllability} discuss how different actuation patterns may affect the accessible set \cite{sontag1995linear} of a soft robot.

Let us assume that the equilibrium $\bar{q}$ is attainable with the given under-actuation matrix $A(q)$.
Under this assumption, then \eqref{eq:feedforward_dynamic} can be generalized in (Fig. \ref{fig:regulation_tracking_block})
\begin{equation}\label{eq:ff_under}
	\tau = A^{\mathrm{L}}(\bar{q})(K(\Bar{q}) + G(\Bar{q})),
\end{equation}
with $A^{\mathrm{L}}$ left-inverse of $A$, as for example the Moore-Penrose pseudoinverse $\left(A^{\top} A\right)^{-1} A^{\top}$. If $A$ is configuration independent, this leads to the same closed loop equation \eqref{eq:dynamics_regulation_ff}. Thus the physical impedance acts as a stabilizing action not only on the collocated part, but also on the variables which are not directly reached by the actuation. If $A$ is configuration dependent then its local changes may have destabilizing effects that must be considered in a modified Eq. \eqref{eq:stablity_condition_ff}, as discussed in the appendix of \cite{della2019exact}.
When dealing with slowly varying trajectories, similar considerations can be applied to the trajectory tracking problem. However, extending the results involving feedback actions - as for example \eqref{eq:PD_trajectory} - is a substantially more complex challenge that is still to be addressed. 
Relying on linearized models can be a practically effective alternative, either when linearizing around the equilibrium \cite{thieffry2019trajectory} or around the desired trajectory \cite{thieffry2020lpv}. 

Control design and analysis get substantially more complex when it comes to stabilizing unstable equilibria of underactuated models. In this case, \eqref{eq:stablity_condition_ff} is not verified, and feedback actions must be necessarily involved. 
A discussion and experimental validation on combining local linear control, an accurate FEM model, and a Luenberger Observer, for designing a damping injection loop is provided in \cite{thieffry2018reduced,thieffry2018control}. A FEM-Based Gain-Scheduling Controller is used in \cite{wu2021fem} to cover the state space of the robot with linear set-point regulators including integral actions.
Moving a step towards the nonlinear domain, the simple controller \eqref{eq:PD_setpoint} can be extended to (Fig. \ref{fig:regulation_tracking_block})
\begin{equation}\label{eq:PD_ua}
    \tau(\bar{q},q,\dot{q}) = A^{\mathrm{L}}(K(\Bar{q}) + G(\Bar{q})) + \alpha A^{\top} (\Bar{q} - q) - \beta A^{\top} \dot{q}, 
\end{equation}
which is a generalization of \eqref{eq:PD_setpoint} to the underactuated domain.
Note that the two gains $\alpha,\beta$ are still elements of $\mathbb{R}^{m \times m}$, and thus they weight the involvement of the actuators into the control loop. 
\begin{Corollary}
	thesis of Corollary \ref{cl:fully_feedback} is verified for the closed loop \eqref{eq:dynamics}-\eqref{eq:PD_ua}, with constant $A$, if the same set of hypotheses obtained is verified when formally switching $\alpha$ and $\beta$ with $A \alpha A^{\top}$ and $A \beta A^{\top}$ respectively, and if
	\begin{equation}\label{eq:atteinable_equilibrium}
		(I - A A^{\mathrm{L}}) (K(\Bar{q}) + G(\Bar{q})) = 0.
	\end{equation}
\end{Corollary}
\begin{proof}
	Under hypothesis \eqref{eq:atteinable_equilibrium}, the following holds $A A^{\mathrm{L}}(K(\Bar{q}) + G(\Bar{q})) = K(\Bar{q}) + G(\Bar{q})$.
	The closed loop dynamics is thus structurally equivalent to the one in Corollary \ref{cl:fully_feedback}, i.e., $M(q) \ddot{q} + C(q,\dot{q})\dot{q} = \left(K(\bar{q}) - K(q)\right) + (G(\bar{q}) - G(q)) + A \alpha A^{\top} (\bar{q} - q) - (D(q) + A \beta A^{\top})\dot{q}$. Thus, the rest of the proof follows as in the fully actuated case.
\end{proof}
The sufficient convergence condition becomes
\begin{equation}\label{eq:stablity_condition_fb_und}
\left. \left(\frac{\partial K(q)}{\partial q} + \frac{\partial G(q)}{\partial q}  + A \alpha A^{\top} \right) \right\rvert_{q = \bar{q}} \succ 0,
\end{equation}
where if $\alpha \succ 0$ then $A \alpha A^{\top} \succeq 0$ but $\mathrm{Rank} \left(A \alpha A^{\top} \right) \leq m < n$. 
Thus, the equilibrium $\bar{q}$ can be stabilized using \eqref{eq:PD_ua} only if the actuation is collocated on the directions in which the effective stiffness loses rank.
Other recent works deal with the regulation of equilibria under similar collocated conditions.
In \cite{franco2021energy} energy shaping controller is proposed for set-point posture regulation one planar segment modeled as a sequence of rigid links, with the same torque applied to all links. 
Moving to more general systems, \cite{boyer2020dynamics} tests in simulation the use of computed torque plus zero-dynamics damping injection in a geometrical exact discrete Cosserat model. This technique was already used for controlling a eel-like hyper-redundant robot in \cite{boyer2006macro}. No proof of convergence is provided, but simulations show good performance.

If also \eqref{eq:stablity_condition_fb_und} cannot be verified, then the problem must be analyzed using less local strategies. 
For example, if the left hand of \eqref{eq:stablity_condition_fb_und} is only semi-positive definite, then the extended version of the Lyapunov function \eqref{eq:lapunov_ff} may still be positive definite.
If however this term is not definite in sign for all $\alpha$, then there are directions on which the actuators is not acting directly, and for which the potential field $K(q) + G(q)$ is repulsive. In this case, stabilization must occur by relying on dynamic couplings.
This is largely an unexplored ground in soft robotics. A very first step in this direction is discussed in \cite{della2020soft}, where a soft inverted pendulum is introduced as an a soft extension of the acrobot \cite{spong1994partial}. The stabilization of an unstable equilibrium is discussed analytically, and it is shown that there is a range of low stiffnesses for which the robot can be stabilized only by means of non-collocated feedback.

\subsection{Actuators dynamics and constraints}

Actuators dynamics plays a much important role in shaping the soft robot behavior, especially if compared to classic rigid robots. Nonetheless, few are the works so far that have explicitly taken into account a dynamics formulation as \eqref{eq:dynamics_a_act} in the design of the controller.
Some actuation technologies require already to accurately consider the control problem for a single isolated actuator. This is the case of electro-thermally-active materials \cite{rizzello2016robust,gu2017survey,van2019closed,copaci2020sma}, and of magnetic actuation of micro and nano robots \cite{pedram2019optimal,jeon2019magnetically}.
If a clear separation exists between the response time of actuators \eqref{eq:dynamics_a_act} and the robot \eqref{eq:dynamics_r_act}, then singular perturbation approach \cite{siciliano1988singular} could be used to improve the performance of the model based controllers introduced above. 
Alternatively, backstepping design achieves the same goal without any assumption on the relative time scales \cite{bridges1995contril}, but at the cost of a more complex control architecture. 
Both techniques have been extensively used to control flexible robots actuated with similar modalities as typically found in soft robotics, as tendon driven \cite{khosravi2014dynamic}, pistons \cite{taheri2014force}, and artificial muscles \cite{asl2017adaptive,carbonell2001nonlinear}.
Nonetheless, the only example of application in soft robotics that we are aware of is a backstepping controller for a single segment approximated with a linear model of the robot and of the air flow \cite{wang2019parameter}.

In soft robotic actuation, it is often the case that the input space can only take values in a subset of $\mathbb{R}^{m}$. This may be due to upper bounds to the maximum force, and to unilateral constraints induced by tendons that can only pull, or pressure chambers that can only push. 
These constraints are usually dealt with heuristics which mask their existence to controllers carefully tuned to not exceed the limits of actuation. As an alternative to heuristics, the masking can also be devised through model based techniques as closed form solution of optimal control allocation problems \cite{della2019dynamic}.
Alternatively, Model Predictive Controllers (MPC) can generate control actions that inherently verify the constraints. 
In \cite{best2016new} linear MPC is used to control a pneumatically actuated humanoid robot, with joint-like localized bending and under a decentralized approximation. 
In \cite{hyatt2020real} the strategy is extended to nonlinear MPC, and Evolutionary algorithms are used to solve the nonlinear optimization. 
In \cite{tonkens2020soft} nonlinear order reduction techniques are used to generate accurate relaxations of a nonlinear finite horizon optimal control problem, including state and input constraints, and formulated on nonlinear FEM models.

\begin{figure}[th!]
	\centering
	\includegraphics[width=\textwidth]{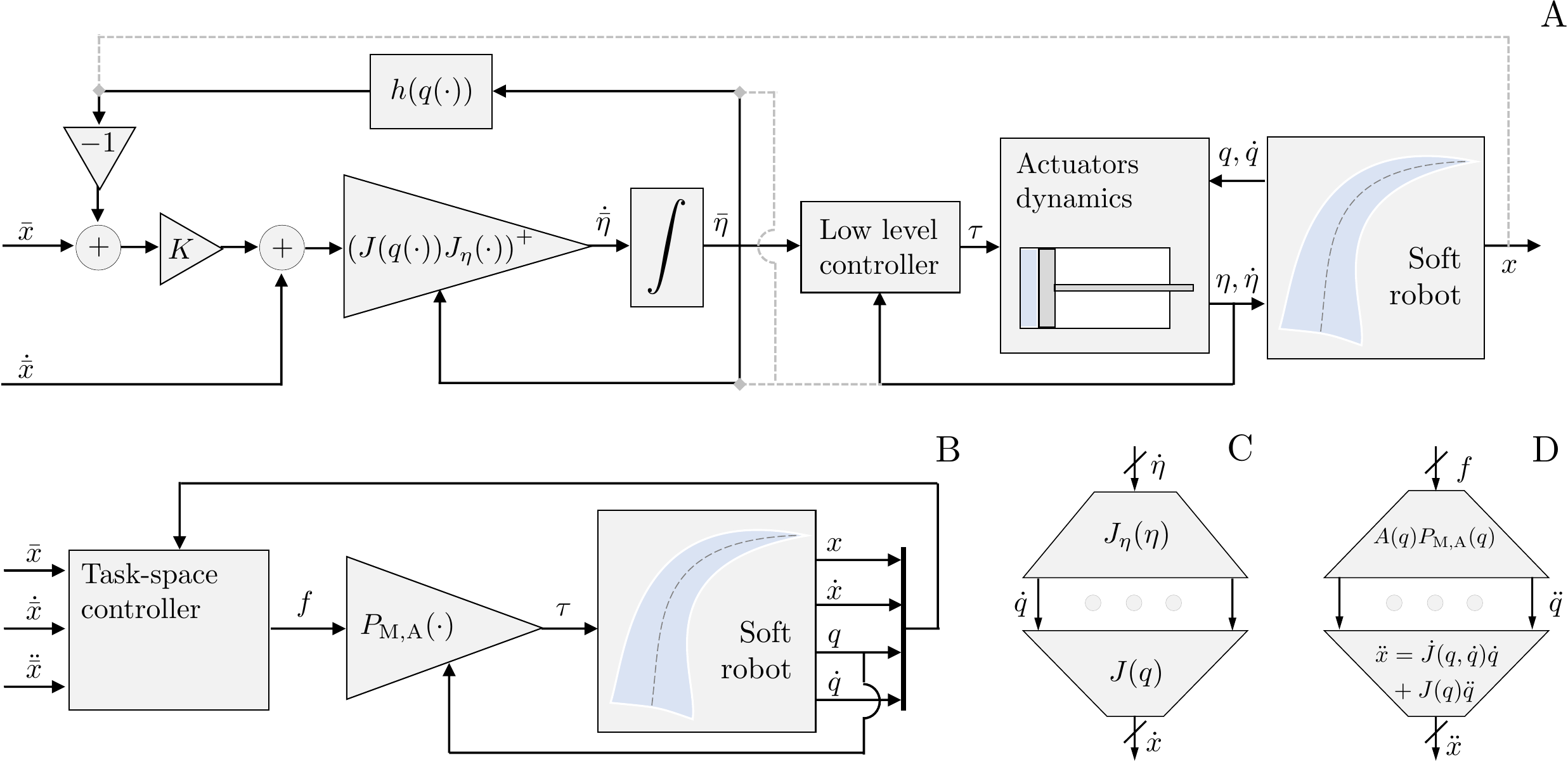}
	\caption{\small  Block schemes of task space controllers in the underactuated case. Panel A shows the standard approach - Eq. \eqref{eq:kinematic_control_under} - which deals with the problem under quasi-static and actuator dominance approximations. Panel B shows a theoretically attractive but potentially unrobust alternative, which acts directly in task space \eqref{eq:ee_dynl}. Both solutions deal with configuration space underactuation by constructing control spaces that are at least as large as the output, but in general, smaller than the configuration. These are actuator side velocities $\dot\eta$ (Panel C) for the first strategy and task level forces $f$ (Panel D) for the second.
	\label{fig:task_block}}
\end{figure}

\subsection{Task space regulation and tracking}

The task space of a robot is usually identified with the configuration of its end effector. In soft robots this corresponds to the configuration of the tip $x(1,t) = h(1,q(t))$. 
For simplicity of notation we will drop the $s$ coordinate in this section. This also allows to stress that the results that we discuss below are general for any $s$ and even for any smooth function $h$ of the configuration $q$. Examples are the potential energy, or the distance of the soft robot from an obstacle.
Thus, we say that a task is fulfilled if 
\begin{equation}\label{eq:task_goal}
	\lim\limits_{t \rightarrow \infty} h(q(t)) - \bar{x}(t) = 0,
\end{equation}
where the desired task coordinates $\bar{x}$ can be either a constant value (regulation) or a function of time (tracking). 
%

A substantial body of literature \cite{bailly2005modeling,camarillo2009configuration,webster2010design,mahl2014variable,bern2017fabrication,bieze2018finite,chikhaoui2018toward} deals with the problem under the kinematic approximation. For a fully actuated model, this means assuming that the robot evolution is described by \eqref{eq:kinematics}, with $\dot{q}$ being the control input. This is a well known problem in robotics \cite{chiacchio1991closed,buss2004introduction,chiaverini2016redundant}, which can be solved with the control loop
\begin{equation}\label{eq:kine_control_fa}
    \dot{q} = J^{+}(q)\left(K_{\mathrm{e}} \left(\bar{x} - h(q)\right) + \dot{\bar{x}}\right),
\end{equation}
with $J^{+}$ being the Moore-Penrose pseudo-inverse of $J$. 
Indeed, combining \eqref{eq:kinematics} and \eqref{eq:kine_control_fa} yields the closed loop dynamics $\mathrm{d}({x} - {\bar{x}})/\mathrm{d} t = K_{\mathrm{e}} ({x} - {\bar{x}})$ that fulfills \eqref{eq:task_goal} exponentially fast for all $K_{\mathrm{e}} \succ 0$.
Note that for $\dot{\bar{x}} = 0$, the time discretization \cite{benhabib1985solution,falco2011stability} of \eqref{eq:kine_control_fa} is equivalent to applying the Newton-Raphson method to solve the following quadratic programming problem
\begin{equation}\label{eq:kine_control_fa_opt}
    \min_{q \in \mathbb{R}^n} \;\; ||h(q) - \bar{x}||^2_2.
\end{equation}
Soft and hard constrains can be explicitly included in \eqref{eq:kine_control_fa_opt}, and possibly reflected in \eqref{eq:kine_control_fa} using multi-task prioritization.
In practice, \eqref{eq:kine_control_fa} is integrated numerically, and the result serves as reference $\bar{q}$ for a low level controller which regulates $q$. This can happen entirely in feedforward or as a high level feedback loop. In the latter case, $q$ and $h(q)$ are directly measured. Alternatively, the kinematic behavior can be forced on the system using based cancellations \cite{kapadia2011task}. Therefore, the use of a kinematic controller implicitly lies on the assumption that all configurations $q$ are attainable through a low level controller as the ones discussed in previous sections. 

To extend \eqref{eq:kine_control_fa} to the underactuated case, one has to introduce some extra assumptions. 
First, it must be assumed that a low level feedback loop $\tau(\bar{\eta},\eta,\dot{\eta},q,\dot{q})$ is available such that if applied to \eqref{eq:dynamics_a_act} then $\eta$ converges to $\bar{\eta}$ fast enough. Under this assumption $\eta$ and $\bar{\eta}$ can be used interchangeably.
This is a strong assumption in general. However, if the robot dynamics is negligible compared to the actuators one - e.g., lightweight robot with strongly reduced actuation - standard actuator-side regulation $\tau(\bar{\eta},\eta,\dot{\eta})$ is sufficient. This is for example the case of lightweight continuum medical devices \cite{burgner2015continuum,da2020challenges}.
Second, it has to be assumed that the robot is drawn to a stable equilibrium $\bar{q}$, whenever a constant $\eta$ is imposed. This is equivalent to say that the feedforward action \eqref{eq:ff_under} generates a stable equilibrium. See the previous subsection for more discussions on the topic.
Third, a one-to-one map must exist from $\bar{\eta}$ to $\bar{q}$, which we refer to as $q(\eta)$. 
We call $J_{\mathrm{\eta}}(\eta)$ the Jacobian of this map.

If these three hypotheses are simultaneously verified, then a differential kinematic model can be constructed, which goes directly from actuators space $\eta$ to task space $x$
\begin{equation}\label{eq:kinematic_control_under}
	\dot{x} = \lefteqn{\overbrace{\phantom{J(q(\eta)) J_{\mathrm{\eta}}(\eta)}}^{\text{Ent-to-end Jacobian}}}
	J(q(\eta)) \underbrace{J_{\mathrm{\eta}}(\eta) \dot{\eta}}_{\dot{q}}.
\end{equation}
%
This is formally equivalent to \eqref{eq:kinematics} from a mathematical standpoint. Thus, a kinematic controller can be constructed by following the same line of reasoning of \eqref{eq:kine_control_fa}, resulting in the control action (Fig. \ref{fig:task_block} A)
\begin{equation}\label{eq:kine_control_ua}
    \dot{\eta} = (J(q(\eta)) J_{\mathrm{\eta}}(\eta))^{+}\left(K_{\mathrm{e}} \left(\bar{x} - h(q(\eta))\right) + \dot{\bar{x}}\right).
\end{equation}
This formulation is quite powerful since $J(\eta)J_{\mathrm{\eta}}(q(\eta))$ is in general full rows rank even if $J_{\mathrm{\eta}}(q(\eta))$ is a strongly higher rectangular matrix (strong underactuaiton of the state), as soon as the dimension of $x$ is smaller or equal than $m$. This is for example the case of a long soft tentacle (Fig. \ref{fig:rod_models}) being actuated with three tendons, and controlled to reach a goal location with the tip. This condition is visually represented by Fig. \ref{fig:task_block} C.
Finally, it is worth underlying that similar steps can be followed by bypassing the actuators models, and directly reasoning on \eqref{eq:dynamics}. This can be achieved by focusing on $\tau$ rather than on $\eta$.
In this case $J_{\mathrm{a}}$ can be derived from \eqref{eq:equilibium_nog}. 
A similar loop as \eqref{eq:kine_control_ua} can thus be used to evaluate the control action in \eqref{eq:ff_under}.

Several variations on the kinematic inversion strategies has been proposed in the literature. The Cosserat kinematic model is combined with linearized task space control in \cite{campisano2021closed}, and with sliding mode control in \cite{rucker2013sliding,alqumsan2019robust}. Visual servoing based kinematic PCC model, where the camera looks the robot, is used to devise the closed loop \cite{xu2021visual}. The inverse kinematics problem is tackled for parallel soft robots by relying on rigid link discretization in \cite{altuzarra2019position}, on FEM models in \cite{zhang2017visual}, and on Cosserat parallel kinematics in \cite{till2015efficient,black2017parallel}.


As an alternative to the many assumptions required by the kinematic approximation, task space control of under-actuated dynamic models can be directly embedded in the dynamic controller by relying on the operational space formulation \cite{khatib1987unified,della2019exact}. As for classic rigid robots, this can be done by differentiating one more time \eqref{eq:kine_control_fa},
and combining the result with \eqref{eq:dynamics_regulation_ff}. Algebraic manipulations yield the operational or task space dynamics
\begin{equation}\label{eq:ee_dynl}
 	\underbrace{\Lambda (q) \, \ddot{x} + \eta(q,\dot{q})  + J^{+\top}_{M}(q) \, (G(q) }_{\text{Terms commonly found in rigid robots}}+ K(q) + D(q) \dot{q}) = J^{+\top}_{M}(q) A(q) \tau,
\end{equation}
where the inertia matrix in the task space is $\Lambda = (J M^{-1} J\mathrm{^\top})^{-1} \in \mathbb{R}^{m\times m}$, Coriolis and centrifugal terms are collected in $\eta(q,\dot{q}) = \Lambda (J M^{-1}C - \dot{J})\dot{q}$, and $J^{+}_{M} = M^{-1} J\mathrm{^\top} \Lambda \in \mathbb{R}^{n \times m}$ is the dynamically consistent pseudo\--inverse.
Eq. \eqref{eq:ee_dynl} resembles the task space dynamics of a rigid robot, with two differences: the robot's impedance $K(q) + D(q) \dot{q}$ and the task space input field $J^{+\top}_{M}(q) A(q)$. The former does not introduce major differences since in any case the integrability of the potential field is lost in task coordinates.
The latter can be solved in general, since it admits the following right hand side inverse  
\begin{equation}
    P_{\mathrm{M,A}}(q) = \left(J(q) M^{-1}(q) A(q)\right)^{-1} J(q) M^{-1}(q) ,
\end{equation}
for all configurations $q$ such that $J(q) M^{-1}(q) A(q)$ is full rank \cite{della2019exact}.
As a result $\tau = P_{\mathrm{M,A}}(q) J^{\top}(q) f$ generates a fully actuated task space dynamics. Thus, direct extensions of standard operational space controllers \cite{nakanishi2008operational} can be used to ensure that \eqref{eq:task_goal} holds for the full dynamic model \eqref{eq:dynamics} and possibly in presence of strong underactuation ($m << n$). This control strategy is depicted in Figs. \ref{fig:task_block} B,D.
Note however that this is not sufficient to ensure that the full state $(q,\dot{q})$ converges to a steady state. How to design a provably stable task space dynamic controller in presence of underaction remains therefore an open problem. This is a challenge which is far from being solved also for classic articulated robots \cite{dietrich2019hierarchical}.

\subsection{Interaction with the environment}

Due to their inherent compliance, soft robots promise to revolutionize how robotic systems interact with the environment by bringing into the picture a new level of safety and robustness compared to standard rigid robots. Yet, most of the works on soft robot control deal with soft robots moving in free space. Moreover, planning algorithms are usually devised so to explicitly avoid any interaction with the environment \cite{roesthuis2016steering,godage2012path,marchese2014whole,della2019dynamic}. 
In practice, the controllers discussed above appear to work well also when interactions with a passive environment occur. Yet, the literature analyzing interactions between the robot and an unstructured environment from a model based perspective is limited. 

Assume that the soft robot is interacting with the environment through a contact area, as for example in Fig. \ref{fig:interaction}. Then, a single point $c$ can be identified called the contact centroid \cite{bicchi1990intrinsic}, such that the net effect of the contact pressure distributions is an equivalent wrench $f_{\mathrm{ext}}, \tau_{\mathrm{ext}} \in \mathbb{R}^{3}$ in $c$. This can be included in \eqref{eq:dynamics} as follows
\begin{equation}\label{eq:dynamics_fa}
	M(q)\ddot{q} + C(q,\dot{q})\dot{q} + G(q) + D(q)\dot{q} + K(q) = 
	\begin{bmatrix}
	    A(q) & J_{\mathrm{c}}^{\top}(q)
	\end{bmatrix} 
	\begin{bmatrix}
	    \tau \\
	    f_{\mathrm{ext}} \\
	    \tau_{\mathrm{ext}}
	\end{bmatrix},
\end{equation}
where $J_{\mathrm{c}}(q) \in \mathbb{R}^{6 \times n}$ is the Jacobian mapping $\dot{q}$ to the linear and angular velocities of the robot in $c$.
A way of characterizing interactions is to look at the Cartesian stiffness matrix, which quantifies the change of reaction forces as a result of a perturbation of the contact location.
In unloaded conditions \cite{salisbury1980active,chen2000conservative} the physical Cartesian stiffness generated by the robot's softness is
\begin{equation}\label{eq:cartesian_stiff}
	K_{\mathrm{x}}(q) = J_{\mathrm{c}}(q)\underbrace{\left(\frac{\partial K(q)}{\partial q} + \frac{\partial G(q)}{\partial q}\right)}_{\text{Stiffness in configuration space}} J_{\mathrm{c}}^{\top}(q) \in \mathbb{R}^{6 \times 6}.
\end{equation}
Note that \eqref{eq:stablity_condition_fb} requires that the stiffness in $q$ space is as high as possible for maximum open loop stability. On the contrary, \eqref{eq:cartesian_stiff} requires to keep it small if the robot is required to behave compliantly in interaction with the environment.
There is therefore a trade-off between softness and stability which must be carefully considered during the robot's design phase. One way to resolve it is to consider ways of changing joint stiffness in time.
Configuration space stiffness can be varied actively by relying on feedback control \cite{suzumori1992applying,della2017controlling}, or passively by changing physical properties of the system \cite{gillespie2016simultaneous,mustaza2018stiffness,rizzello2019towards}.
The inversion of \eqref{eq:cartesian_stiff} is investigated in \cite{albu2004soft}, and in \cite{ajoudani2014natural,bern2021soft}. In the first authors discuss how to prescribe stiffness in configuration space so to achieve a desired Cartesian stiffness, while in the latter is the configuration $q$ to be optimized.
Direct Cartesian impedance control schemes have been proposed and experimentally validated by relying on the kinematic approximation \eqref{eq:kine_control_ua} in \cite{mahvash2011stiffness}, and on the task space dynamic formulation \eqref{eq:ee_dynl} in \cite{della2020model}. 
An in depth introduction to Cartesian impedance control for flexible systems is provided in \cite{ott2008cartesian}.
As an alternative, the wrench $f_{\mathrm{ext}},\tau_{\mathrm{ext}}$ can be directly regulated using Cartesian force control loops. This is achieved in \cite{bajo2016hybrid,yasin2021joint} under the kinematic approximation \eqref{eq:kine_control_ua}.
In \cite{zhang2019motion} 
control inputs are numerically evaluated as the ones minimizing a weighted sum of interaction forces and error at the end effector, and relying on a quasi-static FEM model. 
A similar strategy has been used to implement whole body manipulation \cite{coevoet2019soft}, when a model of the environment is available.

\begin{figure}[th!]
	\centering
	\includegraphics[width=.575\textwidth]{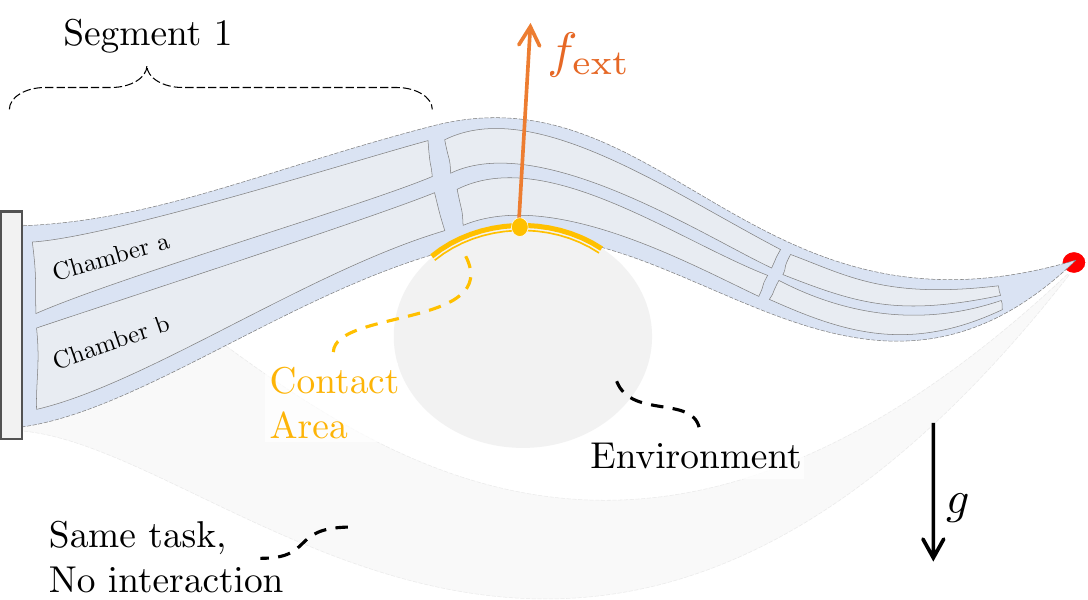}
	\caption{\small  Pictorial example of a soft robot composed of three pneumatically actuated segments and its relation with a simplified environment. The robot can achieve its tip positioning goal in two ways (desired configuration shown in red). It can plan its actions to avoid the environment altogether, or it can exploit the environment. In the latter case, the control design is more complex since it deals with parallel and possibly hybrid dynamics. On the other hand, the force $f_{\mathrm{ext}}$ exerted by the environment at the centroid of contact (shown as a yellow circle) relieves chambers a and b from the burden of sustaining the robot against gravity. If the contact is correctly preserved, it will also increase the stability margins of the system.
	\label{fig:interaction}}
\end{figure}

Contrary to standard robots, soft robots may need to actively seek interactions with the environment (Fig. \ref{fig:interaction}). Indeed, external wrenches may be seen as an extra actuation source, as it appears evident from \eqref{eq:dynamics_fa}. 
Thus, interactions can be used to overcome the limitations imposed by underactuation ($\mathrm{Span}\left(A\right) \subset \mathbb{R}^{n}$) and input saturations ($||\tau|| < c_{\tau}$). 
The use of external wrenches to sustain the robot's body is called bracing \cite{book1985bracing}, and it is demonstrated with a soft robot in \cite{marchese2016dynamics}.
Alternatively, environmental interactions can be used to enlarge the accessible space. A planning method for vine robots which finds the sequence of interactions necessary to reach the desired locations is discussed in \cite{selvaggio2020obstacle}. 

\begin{figure}[th!]
	\centering
	\includegraphics[width=.5\textwidth]{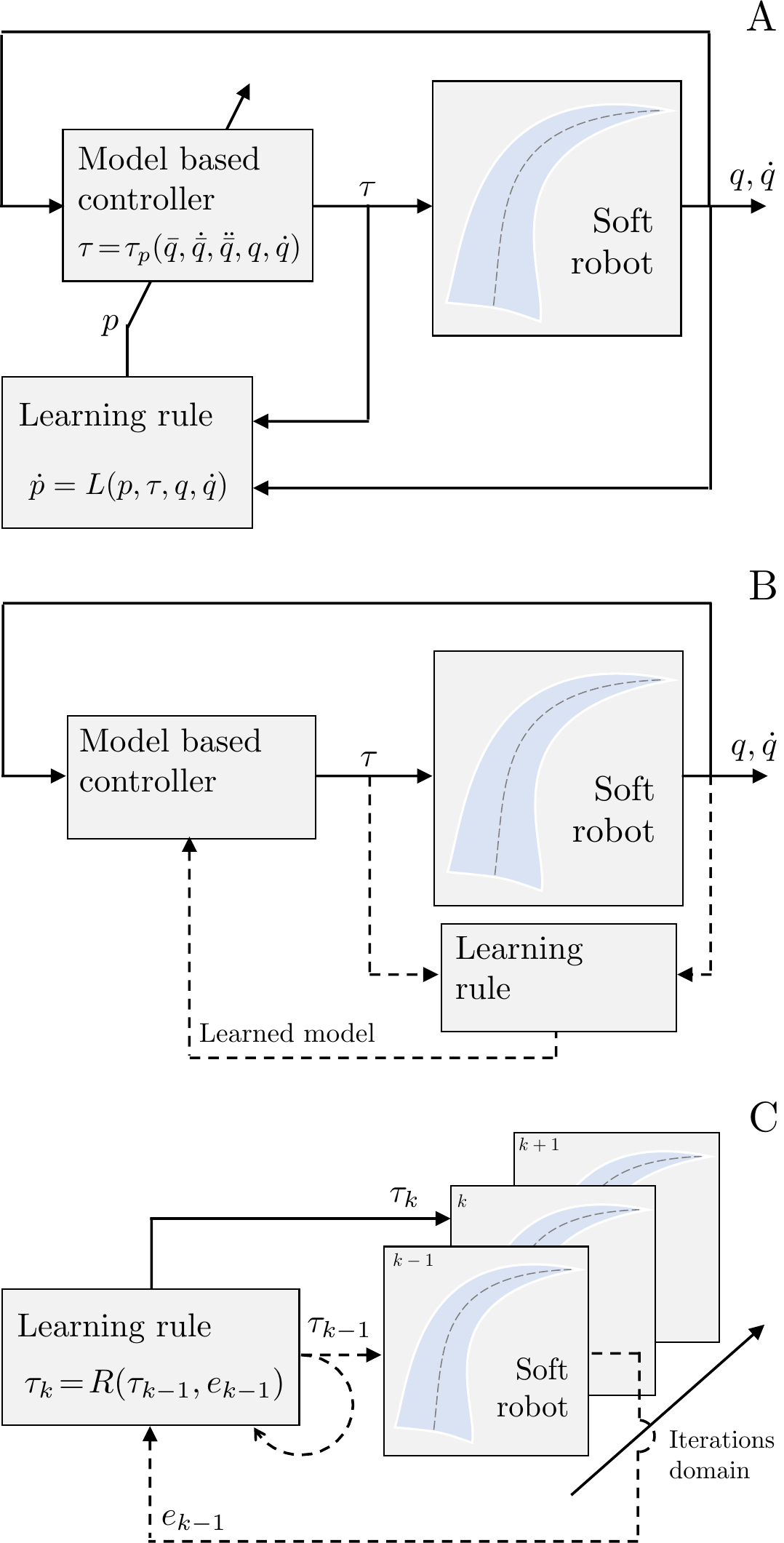}
	\caption{\small  Block schemes of three standard integration solutions between model based controllers and learning rules. Dashed lines represent a transfer of information that happens on a different time scale. Panel A shows an adaptive control architecture, where a learning loop $L$ continuously updates a model based controller. In Panel B a model is learned beforehand. The controller is designed once for all based on the learned model. Panel C reports a standard iterative learning control loop, where a feedforward action is updated iteration by iteration. Here, the learning rule itself $R$ is designed through model based techniques. \label{fig:learning}}
\end{figure}

\section{When first principle models alone are not enough: \\ leveraging data and machine learning in model-based control}

As already mentioned in the introduction, machine learning has been intensively used in the control of soft robots. 
Although many advancements have been made in model based formulations, the importance of integrating data into a model based perspective cannot be understated, especially in the soft robotics context.
At this point of the survey it will probably not come as a surprise that the main reasons are: (i) difficulties in obtaining a complete model (e.g. the actuators cannot be model from first principles, an accurate discretization would be too computationally expensive, the environment cannot be known in advance), and (ii) uncertainties are inherit in any soft robotic application (e.g. unreliability of sensors and actuators, change of physical parameters over time and over several iterations of the same device).

This section focuses on how learning can be integrated into a model-based framework to tackle the control of soft robots. 
Combining model with data is a quite active topic in the control community, and many solutions are currently being developed that will most probably eventually find useful application in the soft robotic field \cite{nageshrao2015port,bucsoniu2018reinforcement,de2019formulas,berberich2020data,sznaier2020control}. 
It is however beyond the scope of this work to discuss these new advancements.
Moreover, the same special issue contains another survey paper fully focused on machine learning strategies for soft robots \cite{cecilia}.

\subsection{Using models to drive learning}

The acquisition of new information and its transformation into a control action can be driven by the knowledge of (an approximation of) the model itself. This can be done while learning a feedback control, a feedforward action, or by serving as source of synthetic data for more standard machine learning approaches.

\subsubsection{Adaptive control}

Adaptive control (Fig. \ref{fig:learning} A) is an established technique in control theory \cite{aastrom2013adaptive}, which allows to augment feedback controllers with an online learning loop. 
The typical structure of an adaptive controller is
\begin{equation}\label{eq:adaptive}
	\dot{p} = \mathcal{L}(p, \tau, q, \dot{q}), \quad \tau = \tau_{p}(\Bar{q},\dot{\Bar{q}},\ddot{\Bar{q}},q,\dot{q}).
\end{equation} 
The uncertainty is represented as a set of unknown parameters $p \in \mathbb{R}^{o}$ appearing linearly in \eqref{eq:kinematics} or \eqref{eq:dynamics}. 
The control action $\tau$ is generated through a model based controller $\tau_{p}$ parametrized in $p$. 
For example, the dynamics of the constant curvature segment discussed in sidebar \ref{sb:CC_derivation}\ref{?} can be linearly parametrized in $m L^2$, $m g L$, $\int_0^1 k(s) \mathrm{d}s$, and $\int_0^1 s \, d(s) \mathrm{d}s$. The second and the third would for example appear as parameters in an adaptive version of \eqref{eq:PD_setpoint}.
The model structure can guide the design of a learning rule $\mathcal{L}$, such that $p$ is moved towards values that better explain the data. 
If $\mathcal{L}$ learns the parameters $p$ that describe the real system, then from there $\dot{p} = 0$ and $\tau_{p}$ behaves as a standard model based controller. 

Classic results in adaptive PD \cite{tomei1991adaptive} and PD+ \cite{slotine1987adaptive} control can be potentially applied to the soft robotic case, by leveraging the equivalence exemplified by \eqref{eq:dynamics_regulation_ff}. Yet, the transfer is less direct than for the non adaptive case. 
An adaptive visual servoing scheme is proposed in \cite{wang2016visual}, by relying on a kinematic PCC approximation. 
Linear adaptive control is used in \cite{skorina2017adapting} to control a single bending actuator.  
A similar strategy is applied to the control of a soft hand exoskeleton in \cite{tang2019model,tang2021model}, where also the parameters of the human fingers are learned.
An adaptive version of MPC is discussed in \cite{terry2019adaptive,hyatt2020model}, and experimentally validated against a non-adaptive MPC. 
A nonlinear adaptive controller originally developed for rigid robots is applied to soft robots in \cite{trumic2020adaptive}, by relying on an augmented rigid robot approximation. 
Adaptive control can be extended beyond known-unknowns by including in $p$ generic disturbances acting on the system. 
High gain observers are used in \cite{zheng2020control} to estimate and compensate for the mismatch between a simplified linear model of a parallel soft robot and the real system.
Also, \cite{franco2021position} discusses the position regulation in Cartesian space for a soft robot with discretized but uncertain kinematics, including adaptive compensation of disturbances.
An adaptive loop that learns the gains of a sliding mode controller is proposed in \cite{cao2021observer}.
In classic robotics, this concept has been pushed even further by adding nonlinear black box approximators borrowed from machine learning to the original dynamics, and including their weights in the $p$ vector \cite{sanner1991gaussian,polycarpou1996stable,yang2017global}.

\subsubsection{Iterative Learning Control}

As an alternative to learning the feedback loop, models can be used to guide the learning of a feedforward action (Fig. \ref{fig:learning} C). This can be done under the hypothesis that a same task can be tried out multiple times by the robot, by using iterative learning control \cite{bristow2006survey} (ILC). Call $k \in \mathbb{N}$ the iteration index, which captures how many times the robot has attempted the execution of the task. Then, in this context a learning rule $R$ is a way of updating the feedforward action
\begin{equation}
    \tau_k(t) = R(\tau_{k-1}([0,t_{\mathrm{f}}]), e_{k-1}([0,t_{\mathrm{f}}])), 
\end{equation}
where $e(t)$ is a measure of how well the task has been executed at time $t$. Note that the learning rule $R$ can in general combine information from the whole error and control evolution at the step $k - 1$. The design of $R$ is driven by the knowledge of the nominal model, and it is defined is such a way that eventually the robot learns a feedforward action $\tau_\infty(t)$ implementing a perfect execution of the task ($||e_{\infty}|| = 0$).
ILC is particularly suited for soft robots since: (i) the learning process is robust to uncertainties in the model, (ii) purely feedforward actions do not disrupt the physical softness of the system \cite{della2017controlling}, and (iii) they are inherently stable if the robot is not too soft.
The actuation patterns necessary to track an optimal trajectory are learned using this technique in \cite{marchese2016dynamics}. Crawling motions of soft worms have been improved via ILC in \cite{seok2012meshworm,chi2019iterative}.
A linear discrete learning rule is used in \cite{hofer2020design,zughaibi2020fast} to control a soft spherical joint. Nonlinear ILC is used in \cite{cao2020analytical} to control a soft finger.
ILC is combined with MPC in \cite{tang2019novel} and used to control a soft bending actuator mounted on an human finger. A continuous feedback-feedforward rule is proposed in \cite{cenceschi2021pi} and applied to the swing up of a soft inverted pendulum.

\subsubsection{Simulators as virtual environments}

A more indirect way of using models to drive learning is by building simulators. Models developed from first principles can be used to create virtual environments where robots can evolve \cite{rieffel2014growing,cheney2015evolving} or learn new skills \cite{tobin2017domain,rojas2021differentiable,naughton2021elastica}. Differentiable simulators are particularly suited to be used in this context, as discussed in \cite{bacher2021design}. Indeed, having gradients available allows to backpropagate through the simulator, which opens up the possibility of apply supervised machine learning and directly optimize for a desired behavior.

\subsection{Control loops based on learned models}

Directly learning an end-to-end controller using standard machine learning techniques can present several disadvantages, as limited explainability \cite{samek2017explainable} and difficulties in ensuring performance and stability of the closed loop. Moreover, the learning process may be ill conditioned due to highly redundant nature of soft robots. 
An alternative to learning the controller is to learn the model and then using it within a model-based control framework.

\subsubsection{Learning the model, by using the model}
Approximation of dynamical systems is discussed in \cite{sjoberg1995nonlinear,schoukens2019nonlinear}, and model learning for robot control is surveyed in \cite{nguyen2011model}.
When using fully black box approximators, first principle models can still be used to generate a warm start for the learning process.
For example, \cite{johnson2021using} uses a liner model of pneumatic actuator for pre-training a neural network, which is then fine tuned with experimental data. A nonlinear model of multi-segment soft robot is used in \cite{tariverdi2021recurrent} to train a recurrent neural network.
Even when a good model is already available, it can still be the case that learning strategies are used to better its performance.
A method for learning only the nonlinear stiffness characteristics of a soft robot is discussed in \cite{deutschmann2018method}. The acquisition of data is driven by a FEM model of the elastic part. An optimal estimation strategy of geometrical quantities describing the robot kinematics is discussed in \cite{wang2019geometric}.

\subsubsection{Closing the loop}

Once learned, the models can be used as a base for model based control loops (Fig. \ref{fig:learning} B). 
In \cite{bern2020soft} a shallow neural network is used to learn the forward kinematics of a soft robot. Then \eqref{eq:kine_control_ua} is used to solve the inverse problem. Neural networks are fully differentiable, and thus Jacobian can be easily calculated.
This is advantageous if compared to directly learning the inverse kinematics, since \eqref{eq:kine_control_ua} resolves automatically those redundancies that would make the direct learning of an inverse kinematics ill posed.
Visual servoing under kinematic approximation for controlling the shape of a soft object is extended in \cite{lagneau2020active}  to the case where the Jacobian is estimated online.
Similar strategies can be employed also in a dynamic setting.
Whenever \eqref{eq:stablity_condition_ff} is verified, learned models can be used to produce open loop control actions which are ineherently stable \cite{thuruthel2017learning,thuruthel2018stable,satheeshbabu2019open}.
Learned models can be used also when feedback is needed to stabilize the desired behavior.
For example, a neural network is trained to approximate the update function of a soft segment in \cite{gillespie2018learning}, and its input-output gradient is used to run a linear MPC algorithm. A similar strategy is used in \cite{johnson2021using}, but here the neural network is directly incorporated into an MPC control loop.
Alternatively, Koopman theory enables learning directly a linear dynamics evolving within an high dimensional lifted space  \cite{bruder2019modeling,kamenar2020prediction,castano2020control}. The resulting model has been used as a base for LQR \cite{haggerty2020modeling} and linear MPC \cite{bruder2020data}.

\section{Conclusions}
This article has surveyed control strategies for soft robots that rely on model based formulations. Special attention has been devoted to organizing this large body of literature within a coherent framework and with a common terminology inspired by classic robotics and robot control. Thanks to the latter, once the discretization of the infinite-dimensional space is introduced, the similarities between rigid, flexible, soft robots become apparent.  Connections with existing results developed outside the soft robotic could therefore be drawn, and controllers ported from the rigid to the soft continuum world.  On the other hand, using a common language clearly pinpoint the fundamental differences between soft robotics and other related fields. The most apparent one is a large number of degrees of freedom, making soft robots intrinsically underactuated. These characteristics would make the control challenge too complex to be solved, if not in straightforward cases. Nonetheless, the positive definite elastic potential and the strictly dissipative force field that is always present, no matter the discretization, simplify the control problem enormously.

Notwithstanding the significant advancements achieved so far, the research community has barely scratched the model-based view's surface to soft robotics. Many are indeed the challenges that remain open and the questions unanswered. How to take underactuation into account? To which extent the non-actuated dynamics can or can not be neglected? Can generic unstable equilibria be stabilized? How to implement motions which are simultaneously compliant, fast, and precise? And controlled movements involving continuous interactions with an unstructured environment? Can a complete integration of embodied intelligence and control design be reached within the model based framework? And so on.

Finally, it should not be forgotten that soft robotics has been born as an experimental discipline, which aims to revolutionize how robots are entering our lives. Thus, all these theoretical advancements should contribute to realizing this grand vision by endowing real soft robots with unmatched motor capabilities.

\newpage
\clearpage

\clearpage
\section[Kinematics and Dynamics of a Planar Constant Curvature Segment]{Sidebar: Dynamics of a Constant Curvature Segment}
\label{sb:CC_derivation}

The goal of this sidebar is to help a novice in soft robotics to familiarize with the topic by concisely presenting the derivation of the main ingredients of what is arguably the simplest soft robot: a constant curvature (CC) segment. Regardless its simplicity, this case already allows to build many intuitions that can directly generalized to more complex and general cases.
Consider a single planar segment as in Fig. \ref{fig:single_cc}, which is an arc with fixed length $L$ but curvature possibly varying in time.
The scalar curvature $q \in \mathbb{R}$ is sufficient to describe its full configuration. Note that since the curvature is defined here with respect to the normalized arc length, then $q$ is the angle subtended by the arc - also called bending angle. The two concepts have been used interchangeably in this paper.
As a comparison, Fig. \ref{fig:single_cc} reports also the non-continuum element of which a CC segment can be consider the direct extension of: a revolute joint connecting two rigid links of length $L/2$. This can be seen as a rigid-link lumped approximation of the CC segment.

The shape $x(s,t)$ of the soft robot can be expressed by collecting the position and orientation of all the reference frames $S_{s}$ connected to the coordinate $s \in [0,1]$. These quantities can be retrieved via simple geometrical arguments, as visually illustrated by Fig. \ref{fig:single_cc}. The result is
\begin{equation}\label{eq:forward_cc}
    x(s,t) = h(s,q(t)) =  L \begin{bmatrix}
        \frac{\sin{s \, q(t)}}{q(t)} &\frac{1 - \cos{s \, q(t)}}{q(t)} &\frac{s}{L} q(t)
    \end{bmatrix}^{\top}\!.
\end{equation} 
Thus, $q$ can be defined also as the angle between base frame and tip frame. 
Note that $x(s,t)$ has no singularity point, since its limit in the straight configuration ($q = 0$) is well defined and equal to $[L \; 0 
\; 0]^{\top}$. However, the division by $0$ can generate numerical instabilities in practice.
Fig. \eqref{fig:cc_kine} compares how the shape of a CC segment changes compared to the one of its lumped discrete approximation. The two models gets progressively more different with the increase of $|q|$, one reason being that the length arc to which both links of the rigid model are tangent shrinks of a factor $(q/2)\cot{(q/2)}$.

According to \eqref{eq:kinematics}, the  Jacobian matrix mapping the time derivative of the curvature $\dot{q}(t) \in \mathbb{R}$ to  $\dot{x}(s,t) \in \mathbb{R}^{3}$ is
\begin{equation}
    J(s,q) = L\begin{bmatrix}
 \frac{s q\cos\left(s q\right) - \sin\left(s q\right)}{q^2} & \frac{\left(\cos\left(s q\right)-1\right) + s q\sin\left(s q\right)}{q^2} & \frac{s}{L} 
    \end{bmatrix}^{\top}.
\end{equation}
This kinematic description is sufficient to express the inertia according to \eqref{eq:inertia}. If an uniform distribution of mass ($m(s) = m$) and a very thin rod ($\mathcal{J} \simeq 0$) are assumed, then the configuration dependent inertia is
\begin{equation}\label{eq:M_cc}
    M(q) = \frac{m L^2}{20}\,\underbrace{\left(\frac{20}{3}\frac{q^3 + 6 q - 12 \sin\left(q\right) + 6 q \cos\left(q\right)}{q^5}\right)}_{\lim\limits_{q \rightarrow 0} * \, = \, 1} > 0.
\end{equation}
Note that similar closed form solutions for $M$ can be found for different mass distributions and non null inertia. These assumptions are introduced here only for the sake of conciseness.
Fig. \ref{fig:cc_quantities} shows a plot of $M(q)$ for all the curvatures in $[-2\pi,2\pi]$. The inertia decreases with the increase of $|q|$ following a bell curve that goes to $0$ when $|q| \rightarrow \infty$. This is because changes in $q$ are reflected in progressively smaller changes in the shape of the soft robot when the curvature is larger - i.e., $||x(s,q + \delta_q) - x(s,q)||^2_2$ decreases with the increase of $|q|$ for all fixed $\delta_q > 0$.
It is also worth noticing that the inertia of the lumped model with homogeneous distribution of mass is $(m/2)(L/2)^2/3 = mL^2/24$, which is smaller than $M(0)$, despite the two system being perfectly superimposed in the straight configuration. This can be explained by considering that the rigid model neglects the motion of the lower half of the robot, and so an actuation torque sees only the inertia produced by half of the robot's body.

Since $M$ is not constant, this formulation of the CC segment dynamics is affected by the following centrifugal force
\begin{equation}\label{eq:C_cc}
    \begin{split}
        C(q,\dot{q})\dot{q} &= \frac{1}{2} \frac{\mathrm{d} M}{\mathrm{d} t}\dot{q} \\
        &=  
        -\frac{m L^2}{3}\frac{12\,q-30\,\sin\left(q\right)+3\,q^2\,\sin\left(q\right)+18\,q\,\cos\left(q\right)+q^3}{q^6}\dot{q}^2.
    \end{split}
\end{equation}
Note that we could evaluate $C$ by direct differentiation of $M$ since both are scalar. Fig. \ref{fig:cc_quantities} reports the evolution of this force when $q$ changes. As expected from a centrifugal action $-C(q,\dot{q})\dot{q}$ tends to increase $|q|$ for all $\dot{q} \neq 0$.

Consider the base of the robot being oriented with a generic angle $\phi$ with respect to a gravity acceleration of intensity $g$. The gravity potential can be calculated by summing up the contributions of each infinitesimal element 
\begin{equation}\label{eq:Ug_cc}
	U_{\mathrm{G}}(q,\phi) = \int_0^1 \underbrace{m \, g \, \left(x(s,0) - x(s,q)\right)^\top 
	\begin{bmatrix}
	\cos(\phi) \\ \sin(\phi) \\ 0
	\end{bmatrix}}_{\text{Infinitesimal contribution of element $s$}}  \mathrm{d}s,
\end{equation}
which is the variation of the center of mass location with respect to the straight configuration, projected to the direction of the gravity acceleration, and multiplied for $mg$.
According to \eqref{eq:potentials}, direct differentiation of the associated potential yields the gravitational torque
\begin{equation}\label{eq:G_cc}
    \begin{split}
        G(q,\phi) = &-m \, g \left( \int_0^1  J(s,q) \, \mathrm{d}s \right)^\top \, \begin{bmatrix}
        \cos(\phi) \\ \sin(\phi) \\ 0
        \end{bmatrix} \\
           = &-m \, g \, L\left(2 \, \frac{\cos\left(q - \phi\right)- \cos\left(\phi \right)}{q^3} +  \frac{\sin\left(q -\phi\right)  - \sin\left(\phi \right)}{q^2}\right).
    \end{split}
\end{equation}
Fig. \ref{fig:cc_quantities} depicts the case of $\phi = 0$, corresponding to a gravity field aligned with the straight configuration of the robot (pointing downward in Fig. \ref{fig:cc_kine}).
Two relevant symmetries that may help thinking about how $G$ changes with $\phi$ are $G(q,\phi) = - G(q,\phi + \pi)$ and $G(q,\phi) = - G(-q,-\phi)$. 
The flexural rigidity can be modeled as a torque proportional to the local bending of the robot, which is the curvature $q$. Thus, the elastic force is
\begin{equation}\label{eq:K_cc}
	K(q) = \frac{\partial}{\partial q}\overbrace{\int_{0}^{1}  \underbrace{\frac{1}{2} \; k(s) \; q^2}_{\text{Infinitesimal contribution}} \mathrm{d}s}^{U_{\mathrm{K}}(q)} = \underbrace{\left(\int_0^1 k(s) \, \mathrm{d}s\right)}_{\text{Average stiffness}} q,
\end{equation}
where $k(s) \in \mathbb{R}$ is the local stiffness in $s$, which is assumed to be almost constant in order for the CC assumption to hold. Similarly, the damping torque can be evaluated by assuming local dissipation proportional to the variation of curvature. The torque needs then to be mapped in $q$ leveraging the kinetostatic duality 
\begin{equation}\label{eq:D_cc}
	D(q)\dot{q} = \int_0^1 \underbrace{J(s,q)^{\top} \begin{bmatrix}
	0 \\ 0 \\ d(s) \, \dot{q}
	\end{bmatrix}}_{\text{Infinitesimal contribution}} \mathrm{d}s = \underbrace{\left(\int_0^1 s \, d(s) \, \mathrm{d}s\right)}_{\text{Equivalent damping}} q.
\end{equation}
where $d(s) \in \mathbb{R}$ is the local damping in $s$. Thus, both elastic and damping forces are linear under the discussed assumptions.
Equivalent results are obtained also when infinitesimal springs and dampers proportional to the elongation are assumed distributed along the thickness of the robot \cite{della2020model}. Finally, consider the robot to be actuated with a pure torque applied at the tip, resulting in
\begin{equation}\label{eq:A_cc}
	A(q)\tau = J(1,q)^{\top} \begin{bmatrix}
		0 \\ 0 \\ \tau
		\end{bmatrix} = \tau.
\end{equation}
Eqs. \eqref{eq:M_cc}-\eqref{eq:A_cc} can be combined by using \eqref{eq:dynamics}, yielding a scalar second order dynamics for $q$ which has the same structure and structural properties of a lumped joint model with parallel impedance, but with different and more complex expressions. Examples of the resulting evolutions are shown in Fig. \ref{fig:cc_evolution}

\begin{figure}
	\centering
	\includegraphics[height =  0.45\textwidth]{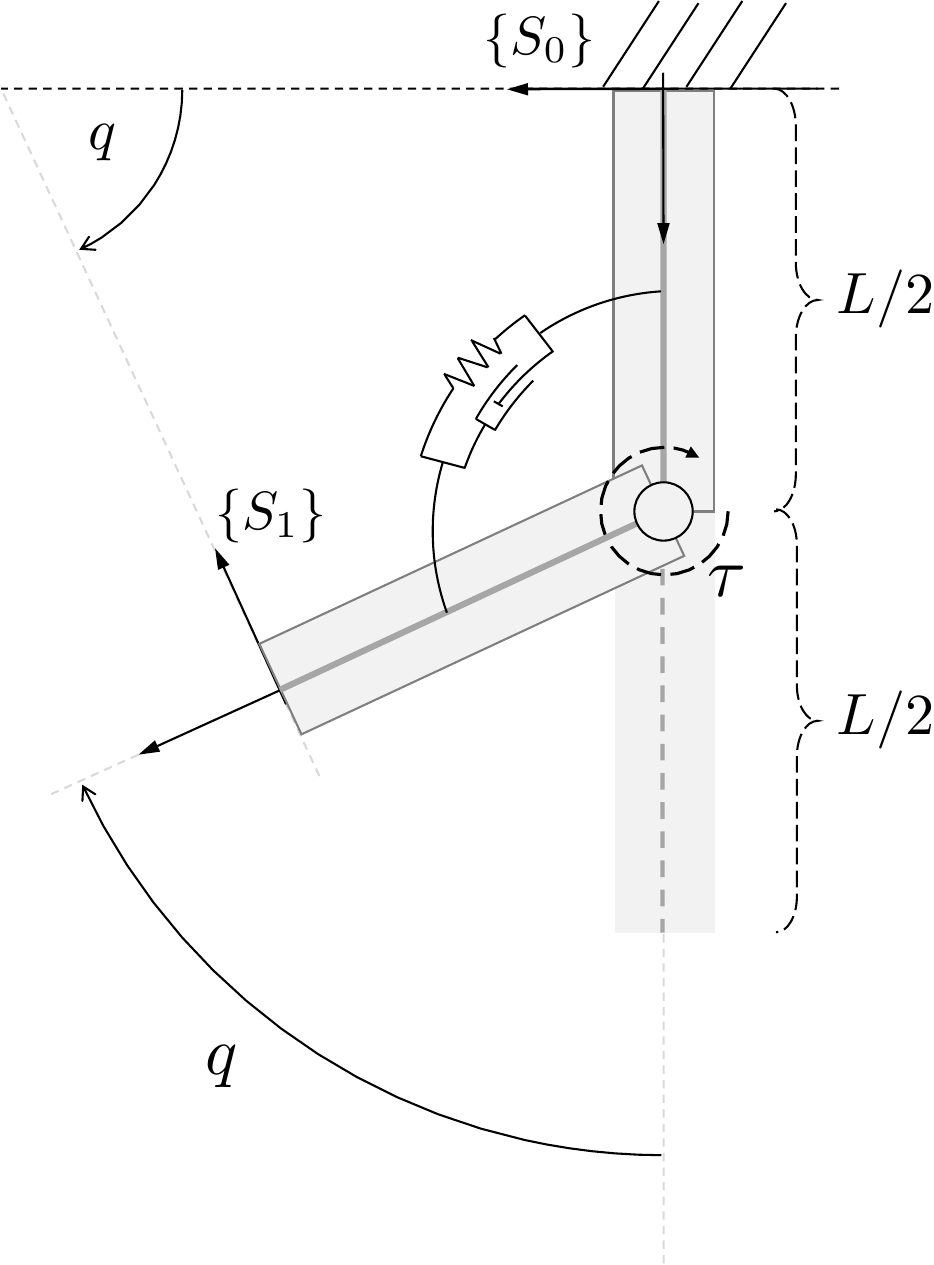} \hspace{0.05\textwidth}
	\includegraphics[height =  0.45\textwidth]{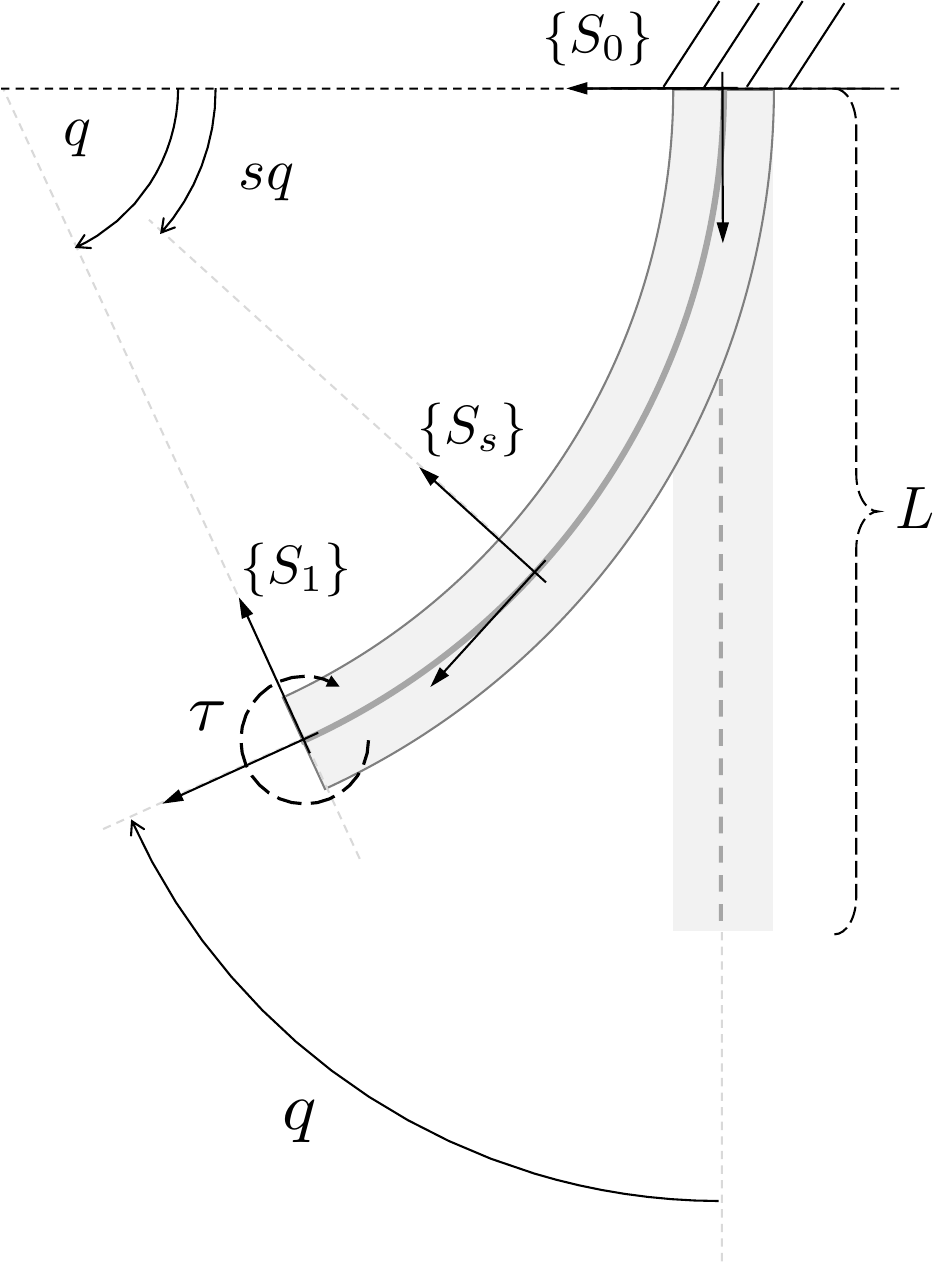}
	\caption{\small A constant curvature segment together with and lumped rigid-link model serving as its first order approximation. The two resulting dynamics have equivalent structural properties, but are described by substantially different dynamic equations. \label{fig:single_cc}}
\end{figure}

\begin{figure}
	\centering
	\includegraphics[width=  0.45\textwidth,trim = {0 0 55 0}, clip]{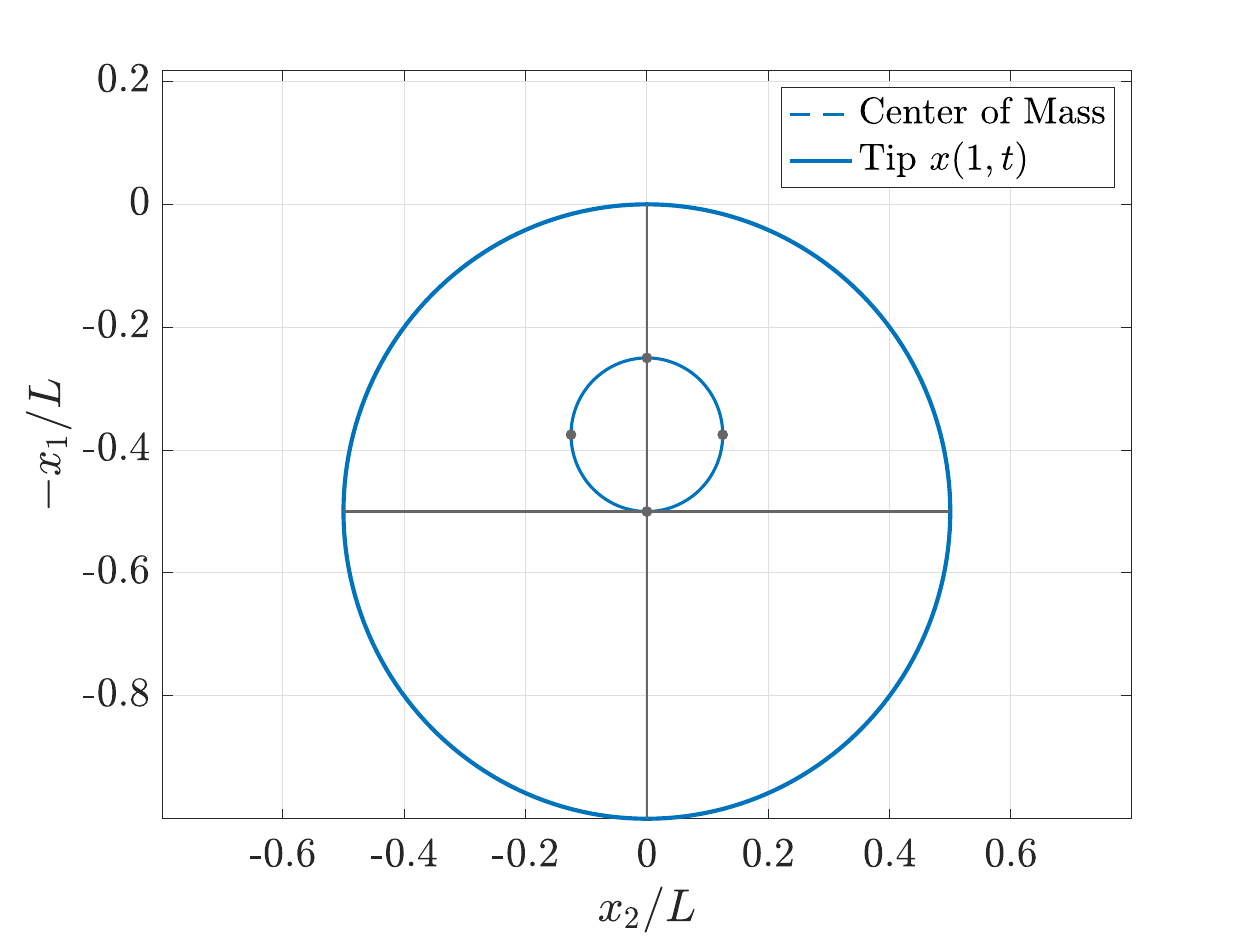}
	\includegraphics[width=  0.45\textwidth,trim = {0 0 55 0}, clip]{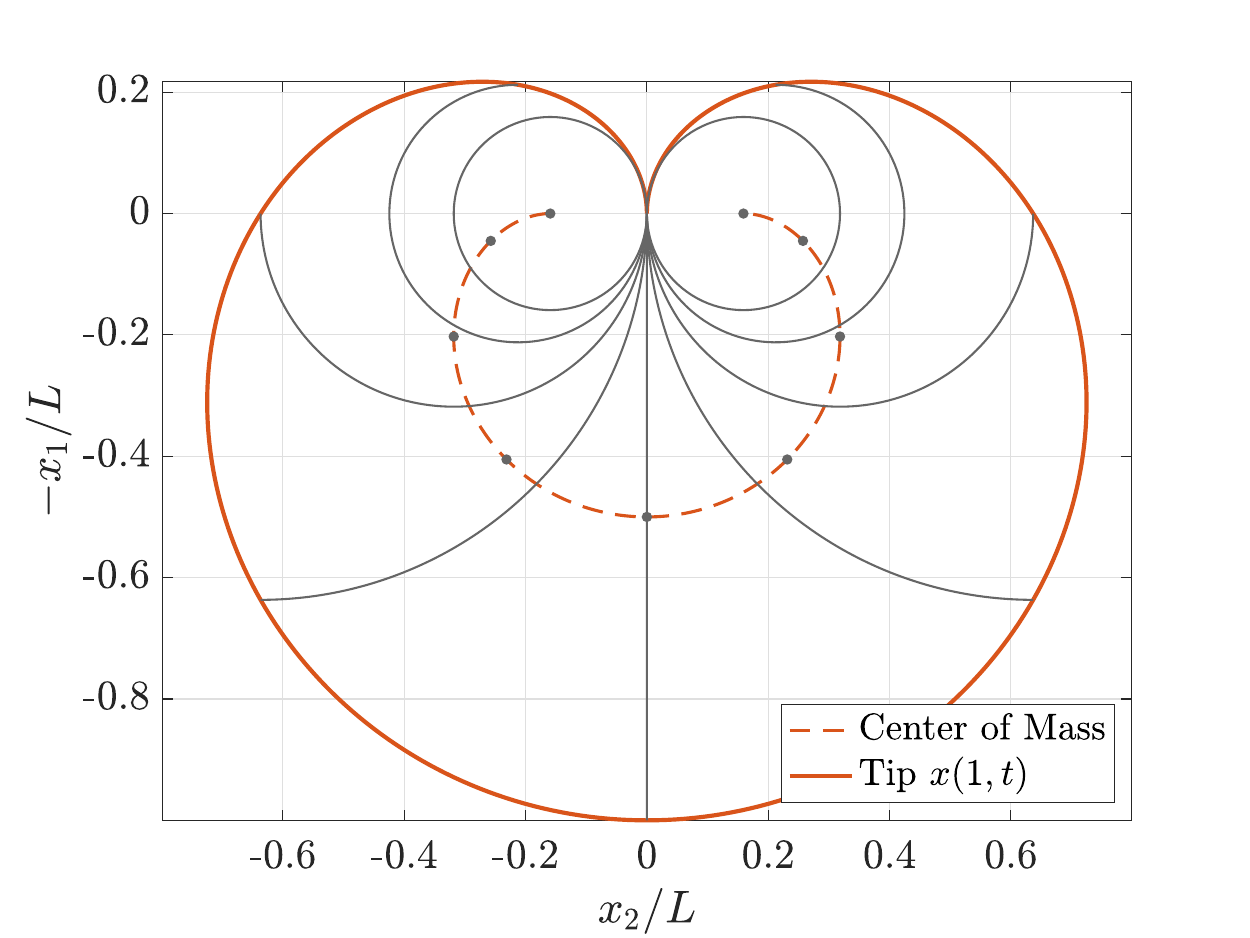}
	\caption{\small Geometrical characterization of a rigid robot with a single revolute joint (left) and of a constant curvature robot (right). the behaviors are similar close to the straight configuration, but strongly depart from each other when the angle $|q|$ increases. The configurations corresponding to $q \in \{-2\pi, -3\pi/2, \dots, 3\pi/2, 2\pi  \}$ are shown with thin gray lines. The corresponding centers of mass are also reported as a gray dot. Note that this range of angles corresponds to two full rotations for the rigid links case. \label{fig:cc_kine}}
\end{figure}

\begin{figure}
	\centering
	\includegraphics[width=  0.3\textwidth,trim = {0 0 45 0}, clip]{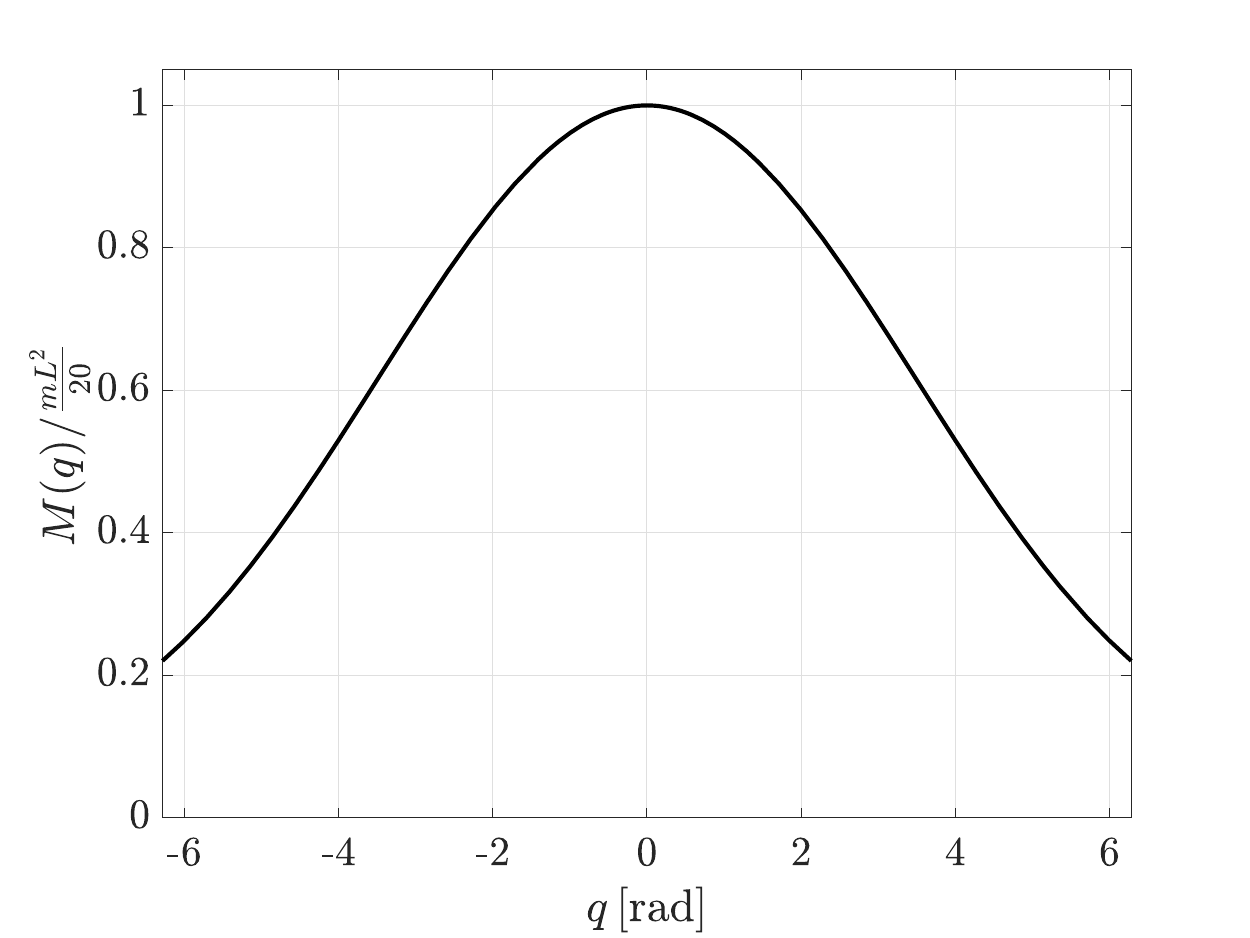}
	\includegraphics[width=  0.3\textwidth,trim = {0 0 45 0}, clip]{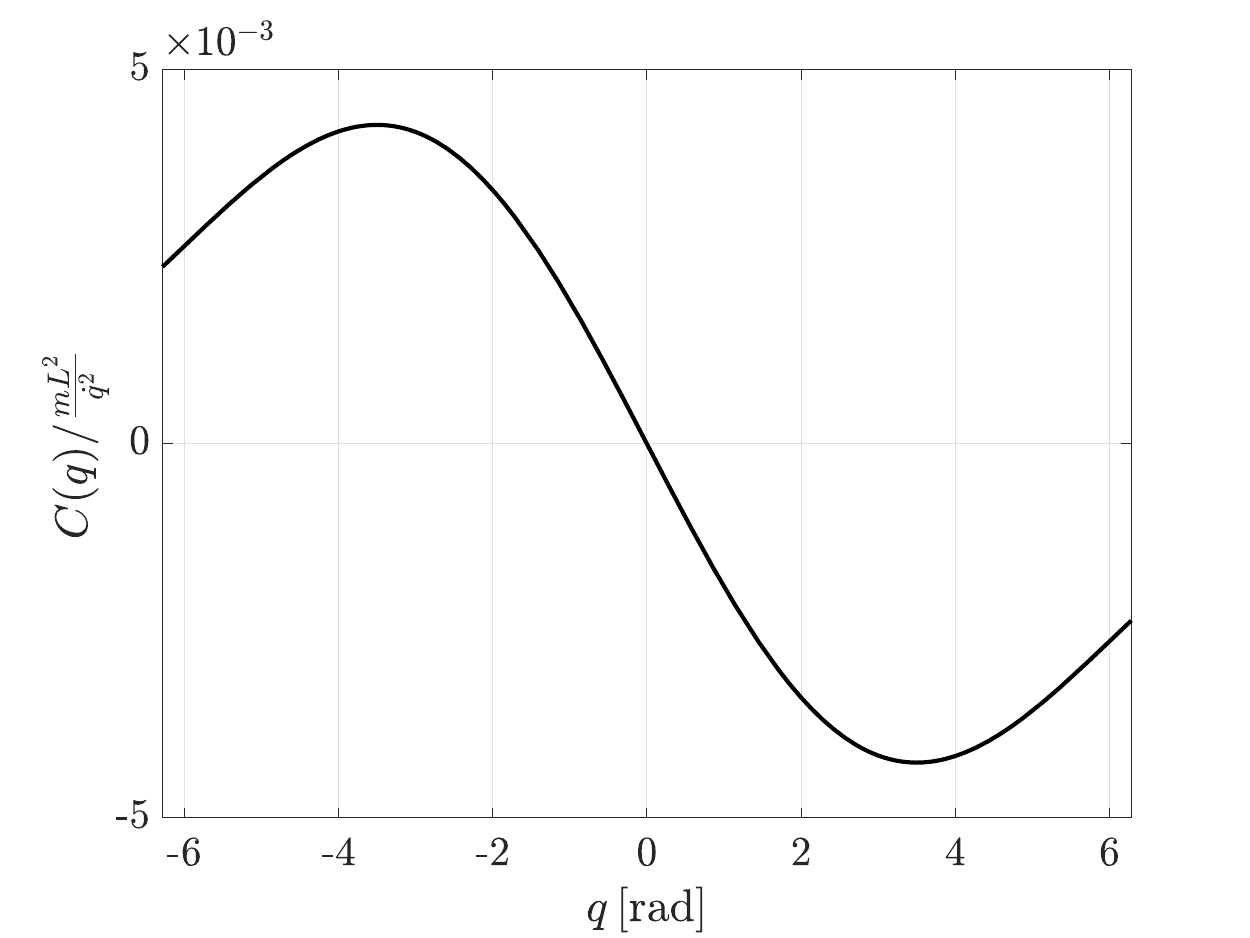}
	\includegraphics[width=  0.3\textwidth,trim = {0 0 45 0}, clip]{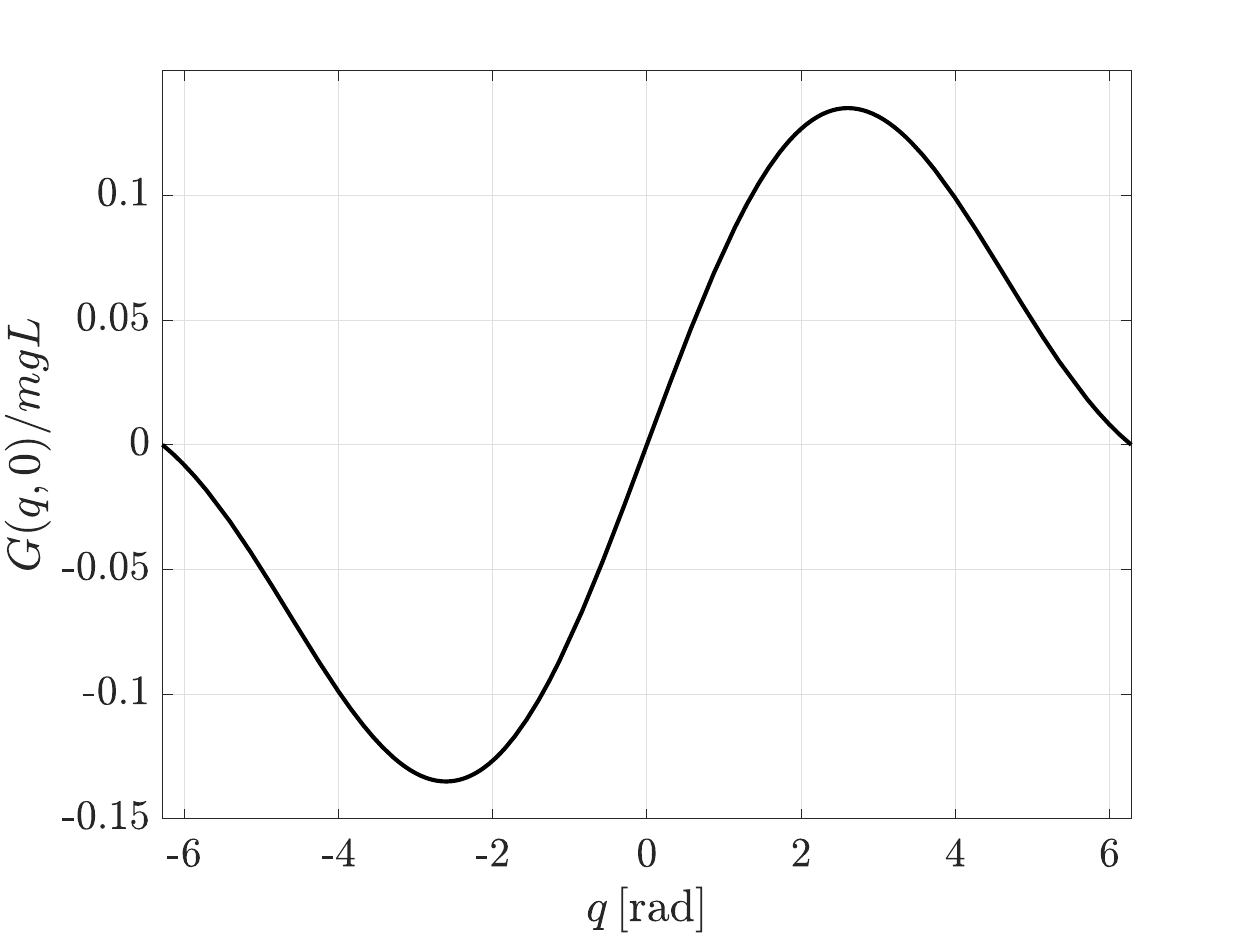}
	\caption{\small Evolutions of \eqref{eq:M_cc}, \eqref{eq:C_cc}, and \eqref{eq:G_cc} all normalized with respect to the quantities that appear linearly in their expression. A change in those quantities result in a linear scaling of the plots along the vertical axis. \label{fig:cc_quantities}}
\end{figure}


\begin{figure}
	\centering
	\includegraphics[height =  0.3\textwidth,trim = {10 0 40 0}, clip]{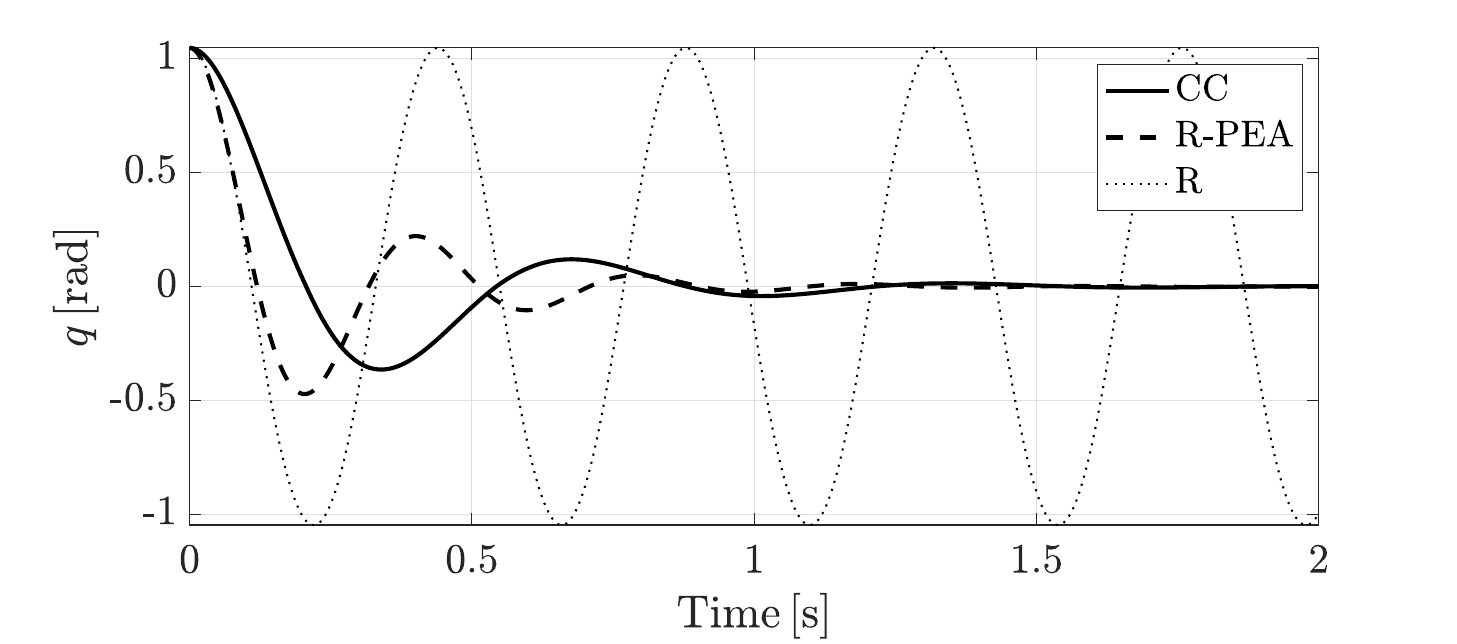}
	\includegraphics[height =  0.3\textwidth,trim = {10 0 40 0}, clip]{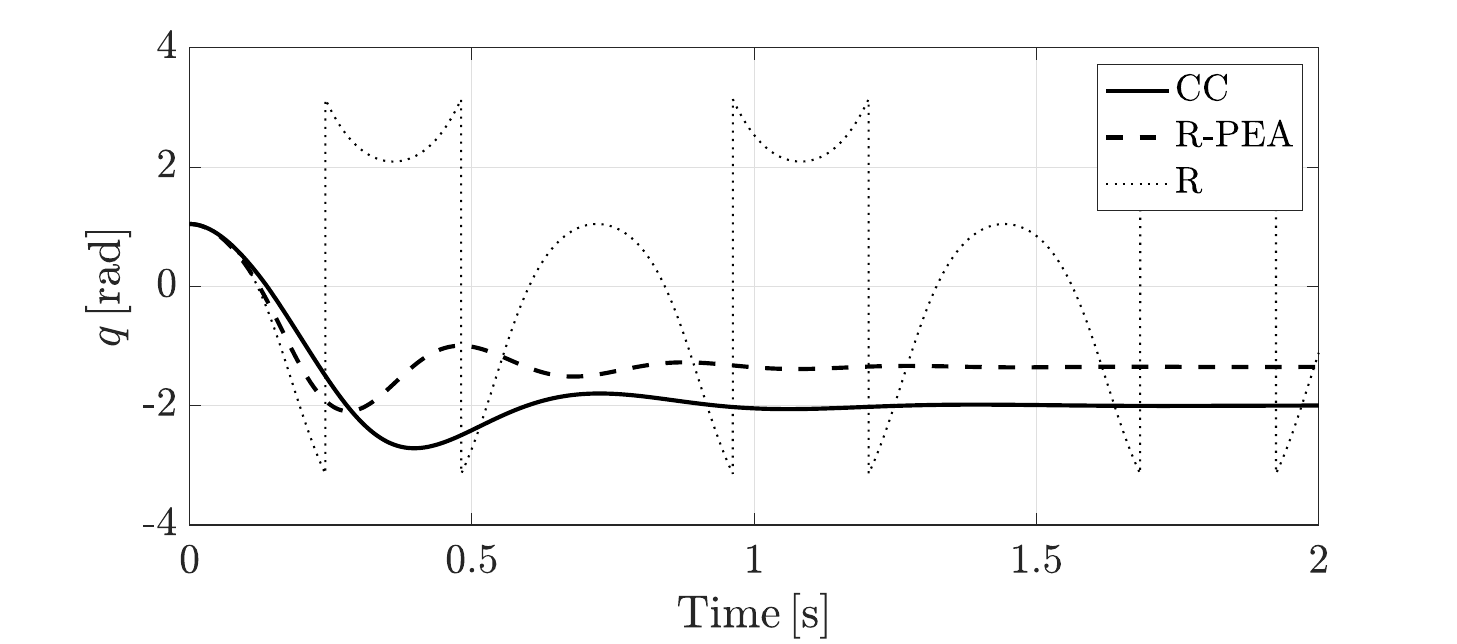}
	\includegraphics[height =  0.3\textwidth,trim = {10 0 40 0}, clip]{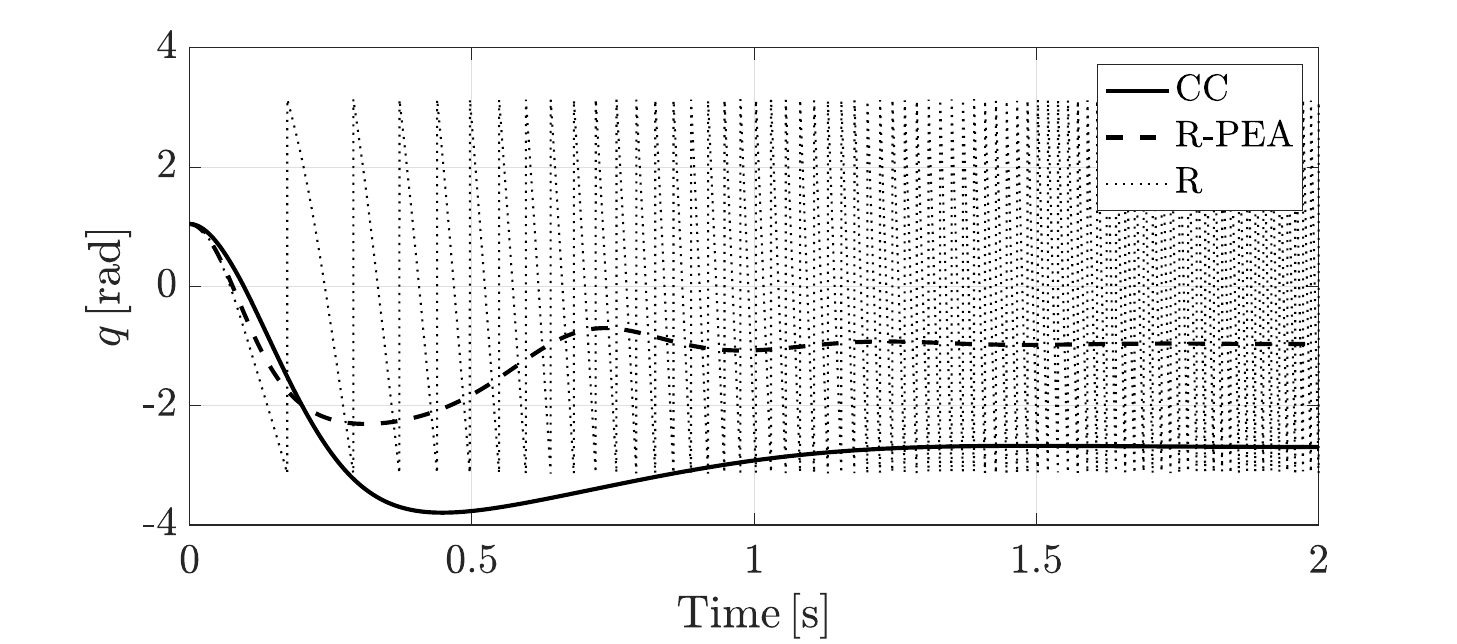}
	\caption{\small Examples of evolution of a constant curvature segment (CC), its lumped rigid link approximation with (R-PEA) and without (R) parallel springs. The CC dynamics is described by \eqref{eq:M_cc}-\eqref{eq:A_cc}. The parameters considered here are $m = 0.5 \mathrm{Kg}$, $L = 0.25 \mathrm{m}$, $\int_0^1 k = 0.05 \mathrm{Nm}$, $\int_0^1 s d = 0.01 \mathrm{N m s}$, and $(q(0),\dot{q}(0)) = (\pi/3,0)$. From top to bottom, the three plots show the evolutions for $(\tau,\phi)$ equal to $(0\mathrm{Nm},0)$, $(-0.3\mathrm{Nm},0)$, and $(0\mathrm{Nm},-pi/2)$ respectively. In all the three cases, CC and R-PEA are qualitatively similar, and both different from R.
	\label{fig:cc_evolution}}
\end{figure}

\newpage
\clearpage

\clearpage
\section[Better than Rigid Robots]{Better than Rigid Robots: Exploiting Softness in Model-Based Control}
\label{sb:better_than_rigi}

A natural way of understanding physical intelligence generated by a soft body within a model based setting is to look at the impedance $K(q) + D(q)\dot{q}$ in \eqref{eq:dynamics} as a low level feedback action. One obvious advantage of implementing such an action physically rather than digitally is that in this way it does not require any additional sensors and actuators. Another important feature is that it acts simultaneously and in a decentralized manner along the whole infinite dimensional structure. 
In other words, $\partial K/ \partial q$ and $D$ are always full rank almost everywhere no matter the level of democratization. The number of independent directions in which standard control can be produced is instead limited, as concisely represented by the fact that the number of columns $m$ of the matrix $A(q)$ in \eqref{eq:dynamics} is independent from the size $n$ of the configuration space $q$.
The consequence in terms of self-stabilization of the robot and control simplification are extensively discussed in the main body of the paper. 

Physical elasticity can also be used to better the execution of dynamic tasks. For example \cite{garabini2011optimality,haddadin2011optimal} prove that lumped join-spring-link systems admit optimal control actions that maximize velocity or forces beyond what can be achieved by a rigid robot of equivalent inertia. During these tasks the potential $U_{\mathrm{K}}(q)$ serves as a tank in which energy can be stored and released when necessary.
Thanks to multi-stabilities and bucklings, continuum structures can lead to even extremer behaviors concerning pick performances, which have be thoroughly investigated in a model based fashion \cite{armanini2017elastica,arakawa2021snap,vasios2021universally}.
However, the proposed descriptions are generally not in the language of dynamical systems, and as such using them for control purposes is still an open challenge.
Another benefit of physical elasticity is to endow the robot with the capability of performing regular oscillations \cite{albu2020review}, which can be excited by means of model based control \cite{garofalo2018passive,marcucci2020parametric,della2020exciting}.
This is especially useful in efficient and robust locomotion \cite{calisti2017fundamentals,kashiri2018overview,bujard2021resonant}.
It is worth underlying that all these capabilities are implicitly exploited by the vast range of approaches using numerical optimization for controlling soft robots \cite{best2016new,marchese2016dynamics,duriez2017soft,bern2019trajectory,hyatt2020real,tonkens2020soft}. 

For an in depth analysis on how the body of a soft robot can generate intelligent behaviors we refer to another paper within the same special issue discussing embodied intelligence \cite{helmut}.

\newpage
\clearpage

\clearpage
\section[Model-based Perception]{Sidebar: Model-based Perception of Shape and Forces}
\label{sb:mb_perception}

Despite the many advances in designing and fabricating soft sensors \cite{wang2018toward}, the perception problem remains an open one in soft robotics.
The use of models can help connecting a finite number of sensor measurements to the virtually infinite amount of degrees of freedom.
Yet, for most of existing models, there exist no sensor capable of directly measuring the configuration space $q$. At the best, a nonlinear combination of the state variables $h(q)$ can be measured, where $h$ is the forward kinematics of the sensor location. This is the dual to the collocation problem that we have encountered in control. Indeed, retrieving a configuration $q$ compatible with the measurements $\bar{x} = h(q)$ is formally equivalent to the task space regulation \eqref{eq:task_goal}, and as such it can be solved by using \eqref{eq:kine_control_fa}.
Alternative kinematic inversion solutions can also be used. For example, constant curvature models admit closed form inverse kinematics \cite{neppalli2009closed}. Nonlinear constrained optimization is used in \cite{hyatt2018configuration} for soft robot with lumped joints.
The knowledge of the robot dynamics can also be taken into account when using nonlinear observers, as the Extended Kalman filter \cite{lunni2018shape,loo2019h}.


The persistence of the potential field $K(q) + G(q)$ allows to connect forces and configurations - especially at steady state - as described by \eqref{eq:equilibrium}. The static inversion of a rigid-link approximations of a soft rod is used in \cite{takano2017real} to extract posture information from a six-axis force/torque sensor place place at the base.
Yet, this relationship is most often used in the other direction: from posture measurements to force sensing. 
Static models can be used to regress an equivalent wrench applied at the end effector from posture information \cite{rucker2011deflection,campisano2020online} under the hypothesis that the robot is lightweight. 
Disturbance observers can be used to detect interactions when the robot mass is not negligible \cite{della2020data}. The location and intensity of the external force can be simultaneously regressed when enough information on the current shape of the robot is available. This is achieved in \cite{bajo2011kinematics} by using a piecewise constant curvature model, and in \cite{chen2020modal} through modal expansion. Numerical inversion of a static Cosserat model can provide an estimation of the whole force distribution through functional expansion of the force profile \cite{aloi2019estimating}.
Static FEM models inversion is used in \cite{navarro2020model} to detect and characterize contacts integrating capacitive and pneumatic sensing. The method can also be applied to soft surfaces. 
Tip forces and robot's shape can also be estimated simultaneously integrating measurements of the base load with a static Cosserat model \cite{sadati2020stiffness}. 

\newpage
\clearpage

\clearpage
\section[Robust control]{Sidebar: Robust control}
\label{sb:Robust}

Models for soft robots always come with some degree of uncertainty.
The nature of the materials used and the manufacturing process are such that the mechanical characteristics of a soft systems can vary dramatically even when starting from a similar original design.
Moreover the state discretization that is at the base of all the discussed control strategies implies that part of the dynamics is ignored.
Both sources of uncertainty will most probably mitigated by advancements in material science and modeling. Yet, these improvements will hardly be sufficient to completely eliminate the issue, which must be taken into account while devising model based strategies. One way of achieving this goal is to integrate learning loops into the controller. Alternatively, control loops can be devised in such a way that they are intrinsically robust to uncertainties.
This is often done implicitly in soft robotics, by avoiding to strongly rely on feedback model cancellations or on high gains. These are indeed characteristics shared by almost all the techniques discussed in this paper.
Alternatively, robustness to uncertainties can be implemented by explicitly relying on robust control design \cite{abdallah1991survey}.
For example, linear robust $H_{\infty}$ control is used in \cite{shu2018robust} 
and \cite{doroudchi2018decentralized} for controlling a single segment and a planar soft robot respectively.
Interval arithmetics is used in \cite{hisch2017robust} to design a nonlinear model based controller which can achieve prescribed tracking performance in presence uncertainties.
Robust sliding model control is considered in \cite{alqumsan2019robust,alqumsan2019multi,cao2021model}.
Fractional order control is used \cite{deutschmann2017robust,munoz2020iso}. The latter is discussed in detail by the survey paper \cite{concha}, part of the same special issue.

\newpage
\clearpage

\clearpage
\section[Infinite Dimensional Control]{Sidebar: Infinite Dimensional Control}
\label{sb:infinite_dimensional}

The question of if infinite dimensional models should be directly used in the design of feedback controllers is a long lasting one, which extends much prior and far beyond the soft robotics field \cite{baker2000finite,jones2010discretization}.
Indeed, working with a finite dimensional approximation of the dynamics substantially simplifies the control design, and already allows to take into account some important features of the robot's dynamics to any desired level of precision.
From a practical standpoint having results which can be proven for any level of discretization ($1 << n_{\mathrm{S}} < \infty$) is de facto equivalent to dealing with the continuum case ($n_{\mathrm{S}} \rightarrow \infty$). Even if appealing, it must be stressed that this approach is not mathematically accurate since it disregards important issues connected to convergence and well-definiteness.
Also, for this line of reasoning to hold the level of discretization of the controller must be kept constant while increasing the discretization of the model. As a simple example, consider the feedforward controller \eqref{eq:feedforward_dynamic}. Different levels of discretization in general imply that the feedforward action does not perfectly match the exact one, i.e. $\tau = \hat{G}(\bar{q}) + \hat{K}(\bar{q})$, with $||\hat{G} - G|| + ||\hat{K} - K|| < \delta$ for some $0 < \delta < \infty$. 
As a result a different equilibrium $\hat{\bar{q}}$ is attained, which is close to ${\bar{q}}$ if $\delta$ is small enough compared to the lipschitz constants of $A$, $K$, and $G$. The robot's configuration converges locally to $\hat{\bar{q}}$ if a version of \eqref{eq:stablity_condition_ff} centered around the new equilibrium is verified.
This analysis becomes more and more complex as soon as non trivial feedback actions are involved \cite{thieffry2018reduced,della2019control}.

On the other hand, even if it requires an arguably substantially more complex formalism, all together avoiding state space discretizations can have two major benefits.
First, it is the only way to exclude that the controller will generate control spillover \cite{balas1978modal,balas1978feedback}. This is a degradation of performance that can eventually bring to instability, due to excitation of high order and otherwise stable dynamics operated by controllers designed using finite dimensional approximations.
Second, infinite dimensional analysis can result in more compact and interpretable solutions compared to the ones based on high dimensional ODEs.
However, the classic theory of PDE control have been mostly focused on linear systems \cite{datko1970extending,pazy2012semigroups,apkarian2018structured}, with extensions to the fully nonlinear case being a topic that is currently being actively researched \cite{mironchenko2020input,rashad2020twenty}.
Consequently, the large majority of applications of PDE control to continuum mechanics \cite[Secs. 4,5]{luo2012stability} deal with systems that for our goals could be consider as a small displacement approximation of the nonlinear rod dynamics \cite{timoshenko1983history}: Euler-Bernoulli and Timoshenko–Ehrenfest beams. These theories study continuum elements undergoing small planar deformations as a result of an external load. Under these assumption, their configuration is described as displacement from a neutral configuration and their dynamics is described by linear PDEs. 
The suppression of vibrations in an Euler-Bernoulli beam subject to boundary actuation can be achieved using local linear feedback \cite{west1987euler,morgul1992dynamic}.
This strategy can be extended to simultaneously verifying constraints in the output by means of barrier Lyapunov function theory \cite{he2015vibration}, and to deal with disturbances and input constrains by using adaptive iterative learning control \cite{he2017adaptive}.
Similarly, linear damping injection can be used to absorb vibrations in a Timoshenko beam subject to boundary \cite{kim1987boundary} or  point-wise \cite{xu2003stabilization} actuation. Damping injection can be combined with energy shaping for configuration control \cite{macchelli2004modeling}. The contact force regulation of a Timoshenko actuated at the base is discussed in \cite{endo2016contact}, under the hypothesis that the environment provides dissipative damping forces

Even if the vast majority of works deal with linear beam models, attention has been also devoted to nonlinear cases.
In \cite{nishida2013optimality} a numerical approximation of an optimal passivity based control is used to stabilize an Euler-Bernulli beam undergoing deformations comparable to the ones of a soft robot.
A practically stable boundary regulator for a nonlinear Timoshenko beam with large deformations is proposed in \cite{do2018stabilization}. 
A boundary feedback control have been proposed in \cite{rodriguez2020boundary} for a beam undergoing large deflections and rotations and small strains, by relying on the fact this system can be mapped to a one-dimensional first-order semilinear hyperbolic system. 
Moving a further step towards the soft robot case we can find works dealing with Kirchhoff rods: \cite{majumdar2013static} discusses the open loop stability of some configurations, \cite{kratchman2016guiding} proposes a purely experimental validation of a kinematic controller, and \cite{bretl2014quasi} proposes a quasi-static manipulation strategy for soft objects by proving that the set of equilibria corresponding to changes in boundary position and orientation constraints is a smooth manifold parametrizable with a single chart.
Finally \cite{chang2020controlling,chang2020energy} use energy shaping and damping injection for posture regulation of soft robots modeled through Cosserat theory and with infinite dimensional input space. Convergence is discussed under finite element approximation. 

\newpage
\clearpage

\newpage
\clearpage

\sidebars 

\section{Authors Biography}

Cosimo Della Santina is Assistant Professor at TU Delft and Research Scientist at the German Aerospace Institute (DLR). He received his Ph.D. in robotics (cum laude, 2019) from the University of Pisa.
He was a visiting Ph.D. student and a postdoc (2017 to 2019) at the Computer Science and Artificial Intelligence Laboratory (CSAIL), Massachusetts Institute of Technology (MIT). He was also a postdoc (2020) at the Department of Mathematics and Informatics, Technical University of Munich (TUM). He is now a guest lecturer at the same university. Cosimo has been awarded euRobotics Georges Giralt Ph.D. Award (2020), and the “Fabrizio Flacco” Young Author Award of the RAS Italian chapter (2019). He also has been a finalist of the European Embedded Control Institute Ph.D. award (2020). His research interests include model based control of soft robots and other elastic systems, combining machine learning and model based strategies with application to mechanical systems, grasping and manipulation.

Christian Duriez received an engineering degree from the Institut Catholique d’Arts et Métiers of Lille, France, and a Ph.D. degree in robotics from the University of Evry, France. His thesis was realized at CEA/Robotics and Interactive Systems Technologies, followed by a postdoctoral position at the CIMIT SimGroup in Boston. He arrived at INRIA in 2006 in the ALCOVE team to work on the interactive simulation of deformable objects and haptic rendering. In 2009, He was the vice-head of the SHACRA team and focused on medical simulation. He is now the head of DEFROST team, created in January 2015. His research topics are Soft Robot models and control, Fast Finite Element Methods, simulation of contact response, and other complex mechanical interactions. All his research results are developed in SOFA, a framework that he co-develops with other INRIA teams. He was also one of the founders of the start-up company InSimo.

Daniela Rus is the Andrew (1956) and Erna Viterbi Professor of Electrical Engineering and Computer Science; Director of the Computer Science and Artificial Intelligence Laboratory (CSAIL); and Deputy Dean of Research for Schwarzman College of Computing at MIT. Rus’ research interests are in robotics, artificial intelligence, and data science. Rus serves as Director of the Toyota-CSAIL Joint Research Center. She is a MITRE senior visiting fellow, serves as a USA expert member for GPAI (Global Partnerships in AI), a member of the board of advisers for Scientific American, a member of the Defense Innovation Board, and a member of several other boards of technology companies. Rus is a Class of 2002 MacArthur Fellow, a fellow of ACM, AAAI, and IEEE, and a member of the National Academy of Engineering and the American Academy of Arts and Sciences. She is the recipient of the 2017 Engelberger Robotics Award from the Robotics Industries Association. She earned her Ph.D. in Computer Science from Cornell University.

\bibliographystyle{IEEEtran}
\bibliography{references}

\end{document}